\documentclass[lettersize,journal]{IEEEtran}
\usepackage{amsmath,amsfonts}
\usepackage{algorithmic}
\usepackage{algorithm}
\usepackage{array}
\usepackage[caption=false,font=normalsize,labelfont=sf,textfont=sf]{subfig}
\usepackage{textcomp}
\usepackage{stfloats}
\usepackage{url}
\usepackage{verbatim}
\usepackage{graphicx}
\usepackage{cite}
\usepackage{upgreek}
\usepackage{booktabs}
\usepackage{makecell}
\usepackage{multirow}
\usepackage{color}
\usepackage{bbding}
\usepackage{pifont}
\newtheorem{theorem}{Theorem}

\newtheorem{remark}{Remark}
\graphicspath{{figure/}}
\begin{document}

\title{
Characterization and Mitigation of Polyphase-Code  \\Artifacts in 5G NR ISAC
}

\author{Xingkang~Li,~\IEEEmembership{Graduate Student Member,~IEEE},
	Shengheng~Liu,~\IEEEmembership{Senior Member,~IEEE},
	Ziguo~Zhong,
	Fanfei~Xu,~\IEEEmembership{Graduate Student Member,~IEEE},
	Qingji~Jiang,~\IEEEmembership{Graduate Student Member,~IEEE},
	Dazhuan~Xu,
	Yongming~Huang,~\IEEEmembership{Fellow,~IEEE}
	
\thanks{
	This work was supported in part by the National Science and Technology Major Project under Grant (No.2024ZD1300200). (Corresponding Author: Shengheng Liu.)
	
Xingkang~Li, Shengheng~Liu, Fanfei~Xu, Qingji~Jiang and Yongming~Huang are with the School of Information Science and Engineering, Southeast University, Nanjing 210096, China, and also with the Purple Mountain Laboratories, Nanjing 211111, China (email: s.liu@seu.edu.cn).

Ziguo~Zhong and Dazhuan~Xu are with the Purple Mountain Laboratories, Nanjing 211111, China.}
\vspace{-2.5em}
}


\maketitle

\begin{abstract}
Target sensing utilizing 5G New Radio (NR) reference signals has emerged as a prominent research direction in both academia and industry. 
However, non-ideal factors in practical deployments exert a significant detrimental impact on target sensing performance, manifesting as artifacts in the range-velocity (RV) spectrum. These artifacts mask weak targets and cause severe false alarms.
To address these challenges, this paper establishes a theoretical model of artifacts and constructs a data-physics-driven deep learning paradigm for artifact mitigation. 
First, the origin of artifacts and their characteristics are theoretically derived. 
These analyses demonstrate that the artifacts are associated with polyphase codes (e.g., Zadoff-Chu sequences) and reveal their characteristics, including periodic extensions in the range domain and spectral spreading in the velocity domain. 
Then, the physical priors of artifacts are formalized as temporal continuity and spatial consistency, informing the design of the training mechanism for the proposed network. 
Guided by these insights, we propose a physics-informed artifact elimination neural network (PIAENet). 
At its core is a multi-frame selective-masked encoder-decoder module, explicitly designed to incorporate the above priors. 
Specifically, temporal continuity is implemented via a multi-frame mechanism to capture features across consecutive RV spectra.
Meanwhile, spatial consistency is realized through a selective masking mechanism to enhance reconstruction of artifact-affected regions. 
Extensive validation is conducted using real-world measured data collected with commercial mmWave equipment. The polyphase-code-related characteristics of the artifacts are experimentally validated. Meanwhile, the experimental results demonstrate that the proposed PIAENet not only effectively reduces the false target count but also improves the detection probability from $\textbf{79.58\%}$ to $\textbf{98.88\%}$.
\end{abstract}


\vspace{-1em}

\section{Introduction}

With the expansion of 5th generation-advanced (5G-A) and 6th generation (6G) use cases \cite{You2023Toward,ITU2022Future}, such as trajectory tracking and intrusion detection, integrated sensing and communication (ISAC) has emerged as a core enabling technology for future wireless networks \cite{Liu2022Integrated,Nuria2024integrated}. To this end, the 3rd Generation Partnership Project (3GPP) is currently undertaking formal standardization for ISAC, covering its deployment scenarios \cite{3GPPISAC1}, sensing procedures, and key performance indicators \cite{3GPPISAC2}. In general, the ongoing ISAC standardization is developed on the basis of existing communication waveforms and protocols \cite{Shi2022Device,Yang2026Hierarchical}, with the aim of realizing sensing functionality with minimal overhead and modification. 

Aligned with this design principle, target sensing utilizing 5G New Radio (NR) signals \cite{Wypich2025passive,Zhang2025Target}, such as the sounding reference signal (SRS), is a focus of academic research and standardization discussions. These reference signals are often generated from polyphase codes, such as Zadoff-Chu (ZC) sequences \cite{Chu1972Polyphase}. Prized for their excellent autocorrelation and low peak-to-average power ratio, they are therefore extensively used in channel estimation and target sensing.
On this basis, a sensing module can be directly retrofitted onto the existing 5G NR communication system \cite{Liu2026Cooperative,Na2025Integrated}.

However, practical deployments may introduce interference due to non-ideal factors, including unpredictable stochastic noise, transmitter-receiver asynchronism, and local oscillator instability \cite{Ding2025Bi,Mao2025Model}. These non-ideal factors degrade sensing performance far more significantly than communication performance. In real-world scenarios, we verified the target sensing functionality with commercial mmWave devices and observed the striped artifacts shown in Fig. \ref{fig:artifacts1}(a). We repeated the experiments with commercial sub-6 GHz devices from different vendors, and the artifacts persisted as shown in Fig. \ref{fig:artifacts1}(b). These experiments confirm that the artifacts are not induced by device-specific impairments. 
Through extensive experiments, we draw the following preliminary conclusions. When polyphase-code sequences are used for sensing, non-ideal factors produce artifacts in the range-velocity (RV) spectrum. These artifacts can mask distant weak targets, cause false alarms, and severely degrade sensing performance. 
To address the artifact-related issues, our work focuses on \textit{how to mitigate polyphase-code artifacts in 5G NR-based sensing without altering the standardized waveforms or protocols}. 

 

\begin{figure*}[]
	\centerline{\includegraphics[width=0.99\linewidth]{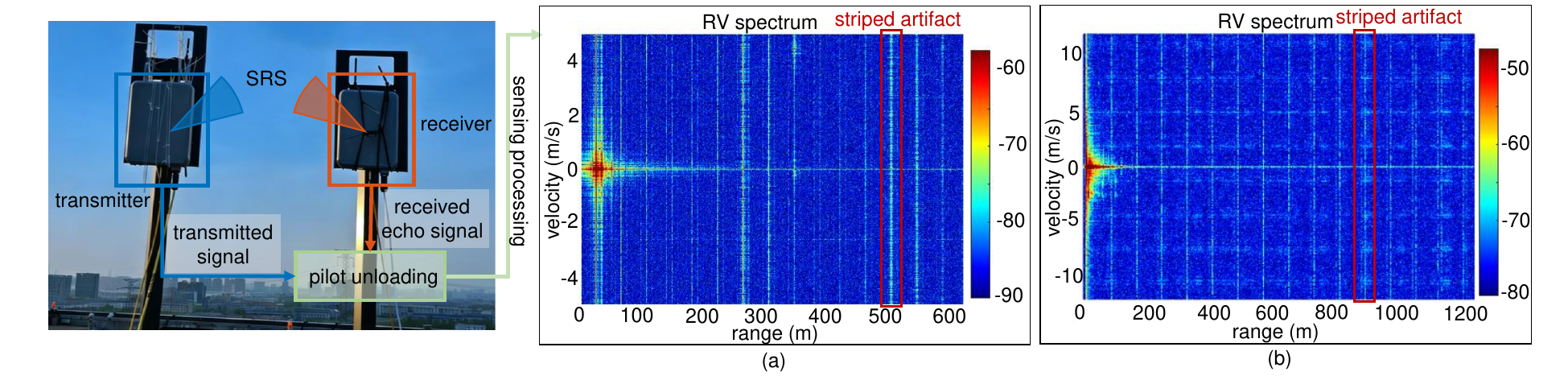}}
	\vspace{-0.75em}
	\caption{{Photograph of the 5G NR-based sensing system and RV spectrum test results. The red bounding box identifies the striped artifacts in the RV spectrum. (a) RV spectrum test results of $25.6$ GHz mmWave devices. (b) RV spectrum test results of $4.9$ GHz sub-6 GHz devices. }}\label{fig:artifacts1}
\end{figure*}

The aforementioned problem can be framed as a specialized interference suppression problem. In the field of radar signal processing, interference suppression methods have been extensively studied. These methods can be broadly divided into three categories, namely physics-driven methods based on theoretical models, data-driven methods based on neural networks, and data-physics-driven methods.

Physics-driven methods are developed based on the physical properties of target echoes. Interference suppression can be regarded as a filtering process in the time-frequency domain \cite{Zhang2024Dual}. For example, wavelet denoising is used to extract interference signals from the output of a time-domain low-pass filter \cite{Lee2021Mutual,Xu2021Interference}. By leveraging the stationary nature of interference in the RV domain, the extended cancellation algorithm can utilize the signal from the reference channel to remove static clutter interference \cite{Alland2019Interference,Jin2019Automotive}. However, weak targets are often completely masked by strong interference, making it difficult to separate the target signal from interference \cite{Baral2023Automotive}.

Deep learning approaches have garnered significant interest for interference suppression \cite{Siam2025Artificial,Nirmal2021Deep}. Convolutional neural networks (CNNs) are widely used to map complex signals to the frequency-domain distribution of target signals, enabling end-to-end training for interference suppression \cite{Ristea2020Fully}. Moreover, the image-like RV spectrum is highly amenable to processing via data-driven methods well established in the computer vision field \cite{Fuchs2020Automotive}. Thus, interference suppression can be regarded as denoising and reconstruction of the RV spectrum \cite{Zhu2021Low,Wang2024Interference}. However, data-driven methods often require manual re-tuning of parameters or even complete redesign when the deployment environment changes. This severely limits the robustness and maintainability of those methods in practical deployments.

Recently, physics-guided deep learning has also been extensively applied to interference suppression \cite{Wang2025Physics,Wang2022Prior}. Such methods typically incorporate physical properties as priors into network training to enhance interpretability, generalization, and robustness \cite{Park2024Interference,Pang2025MFS,Zhang2024FUAS}. For example, by incorporating the spectral-spatial decomposition mechanism as a physical prior, a two-stage generative network is developed to realize high-fidelity signal reconstruction under severe radar interference \cite{Zhao2025RaSPD}. To leverage the motion characteristics of targets, a long short-term memory network is designed to learn the spatial and temporal dynamic characteristics of various targets in the RV domain \cite{Khalid2019Convolutional}. The aforementioned data-physics-driven methods can all achieve interference suppression in their respective specific systems. 
 
Although data-physics-driven methods for interference suppression have garnered significant interest, the novel characteristics of target echoes in 5G NR-based sensing remain to be systematically characterized and exploited for artifact elimination. Therefore, \textit{the polyphase-code artifact mitigation in 5G NR-based sensing}, which is the focus of this paper, requires further dedicated research.


To bridge this research gap, this paper constructs a physics-guided deep learning paradigm to address artifact-related challenges in 5G NR-based sensing.
We first establish a theoretical model that elucidates the origin and characteristics of artifacts. This analysis reveals the distinct physical characteristics of artifacts in the range and velocity domains, which are formalized as physical priors, namely temporal continuity and spatial consistency. 
Guided by these insights, we propose a physics-informed artifact elimination neural network (PIAENet) to mitigate the impact of artifacts on target detection. Within PIAENet, artifact elimination is performed by a multi-frame selective-masked encoder-decoder, whose design is explicitly informed by the derived priors. Specifically, temporal continuity is implemented as a multi-frame mechanism to capture temporal correlations among RV spectra. Meanwhile, spatial consistency is implemented as a selective masking mechanism to enhance the reconstruction of artifact-affected regions. The proposed PIAENet effectively tackles real-world impairments by embedding physical priors into the training mechanism. 

The key technical contributions of this work can be summarized as follows.
\begin{itemize}
	\item We identify the presence of striped artifacts in real-world deployments, which severely degrade sensing performance. When polyphase-code sequences are employed for sensing, practical hardware impairments will induce artifacts in the RV spectrum.
	\item We establish a theoretical model that explains the origin and characteristics of polyphase-code artifacts in 5G NR-based sensing. The analysis reveals that artifacts exhibit periodic extensions in the range domain and spectral spreading in the velocity domain.
	\item We propose a data-physics-driven framework PIAENet centered on a multi-frame selective-masked encoder-decoder module. It incorporates temporal and spatial priors into the network training, enabling effective artifact elimination while maintaining sensitivity to weak targets.
	\item We conduct extensive validation experiments using real-world measured data collected with commercial communication devices. Experimental results show that the proposed network achieves a detection probability of up to $98.88\%$ and effectively reduces the false target count.
\end{itemize}

The rest of the article is structured as follows. Section \ref{sec:SystemModel} introduces the sensing signal model and analyzes the artifact origin. Section \ref{sec:Elimination} details the proposed physics-informed artifact elimination methodology. Section \ref{sec:Experimental} presents real-world experiments and simulations to evaluate the proposed framework. Finally, the conclusion of this paper is drawn in Section \ref{sec:conclusion}.

\textit{Notations:} Lower (upper)-case bold characters denote vectors (matrices), and the vectors are by default in column orientation. The superscript $(\cdot)^{\mathsf T}$ represents the transpose operator. Symbol $\ast$ denotes the convolution operation. $\left\Arrowvert\cdot\right\Arrowvert_2$ denotes the $2$-norm. $\left\lceil\cdot\right\rceil$ represents the ceiling function. $\mathbb{C}$ represents the set of complex numbers. ${\mathrm j}$ represents the imaginary unit. $\mathcal{F}\{\cdot\}$ denotes the Fourier transform.

\section{Signal Model and Artifact Origin Analysis}\label{sec:SystemModel}
For clarity, a bistatic scenario is used to present the work, as shown in Fig. \ref{fig:system}. However, the observation can be extended to monostatic and multistatic scenarios. In 5G NR physical-layer specifications, several reference signals are defined for purposes such as channel estimation. These predefined reference signals can be readily exploited to implement sensing capabilities without incurring significant communication overhead. In the following, we introduce the ideal sensing signal model and systematically elaborate the origin analysis and theoretical modeling of the artifacts.
 
\subsection{Ideal Sensing Signal Model}
In the 5G NR-based sensing system, the transmitter and receiver operate at a carrier frequency $f_{\mathrm{c}}$. The transmitter emits an orthogonal frequency division multiplexing (OFDM) signal for sensing, which consists of $K$ subcarriers with a spacing of $\Delta_\mathrm{f}$ and $L$ symbols. Simultaneously, the receiver performs radar sensing using the echo signals to detect targets and estimate their range and velocity. 
\subsubsection{Transmitted signal model}
The ZC sequence \cite{Chu1972Polyphase} is adopted as the pilot sequence for the SRS \cite{3GPPNR}, owing to its ideal periodic autocorrelation properties, constant-envelope characteristic, and low peak-to-average power ratio. The frequency-domain expression of this SRS pilot sequence on the $k$-th subcarrier is given by
\begin{align}\label{eq:bk}
	S_{k}=\mathrm{e}^{-\mathrm{j} \mathrm{\uppi}\mu k(k+1)/K},\quad k=0,...,K-1 ,
\end{align}
where $\mu$ is the root index of the ZC sequence. 
\begin{figure}[]
	\centerline{\includegraphics[width=0.9\linewidth]{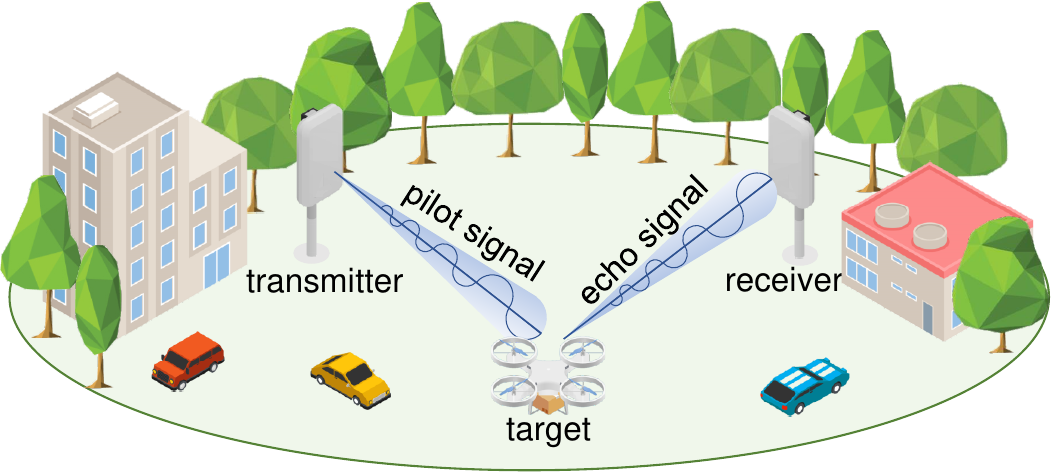}}
	\caption{\textmd{Illustration of a bistatic sensing system using 5G NR communication devices.}}\label{fig:system}
	\vspace{-1.5em}
\end{figure}
The continuous-time baseband signal of a single OFDM symbol at time $t$ is represented as 
\begin{align}\label{eq:sensingsignals1}
	{s}(t)=\frac{1}{\sqrt{K}}\sum_{k=0}^{K-1}S_{k}\mathrm{e}^{\mathrm{j} 2\mathrm{\uppi}k\Delta_\mathrm{f}t}\mathrm{rect}\left(\frac{t}{T_{\mathrm{D}}}\right),
\end{align}
where $\mathrm{rect}(\cdot)$ is the rectangular window function. $T_{\mathrm{D}}=1/\Delta_\mathrm{f}$ is the effective symbol duration. Considering the addition of a cyclic prefix (CP) to the OFDM waveform, the transmitted baseband signal of the $l$-th symbol is rewritten as
\begin{align}\label{eq:sensingsignalx1}
	{x}_l(t)=\begin{cases}
		{s}(t-lT+T_{\mathrm{D}}-T_{\mathrm{CP}}),\, 0\leq t-lT<T_{\mathrm{CP}};\\
		{s}(t-lT-T_{\mathrm{CP}}),\, T_{\mathrm{CP}}\leq t-lT<T;
	\end{cases}
\end{align}
where $T=T_{\mathrm{CP}}+T_{\mathrm{D}}$ is the total duration, including the CP duration $T_{\mathrm{CP}}$. The complete transmitted baseband signal formed by $L$ consecutive OFDM symbols is expressed as $x(t)=\sum_{l=0}^{L-1}{x}_l(t)$. 
Before transmission, the baseband OFDM signal is up-converted to the RF domain.

%
%

\subsubsection{Received echo signal model}
The transmitted OFDM signal ${x}(t)$ is reflected by targets, and the echo signal is then received. After down-conversion at the receiver, the received baseband signal at time $t$ is expressed as
\begin{align}\label{eq:sensingsignaly1}
	{y}(t)=\sum_{q=1}^{Q}\beta_q{x}(t-\tau_q)\mathrm{e}^{\mathrm{j} 2\mathrm{\uppi}f_{\mathrm{d},q}t}+{z}(t),
\end{align}
where $Q$ is the number of targets. $\beta_q$, $\tau_q$, and $f_{\mathrm{d},q}$ are the reflected complex channel gain, delay, and Doppler of the $q$-th path, respectively. ${z}(t)$ is complex Gaussian white noise. It is worth noting that the Doppler shift originating from moving targets is assumed to be identical within a symbol. The received discrete-time signal of the $l$-th symbol at the $n$-th sampling point with sampling interval $T_{\mathrm{s}}$ is written as\footnote[1]{According to the Nyquist sampling theorem and engineering requirements, $T_{\mathrm{s}}$ is typically set as $1/K\Delta_\mathrm{f}$.}
\begin{align}\label{eq:sensingsignalyd1}
	{y}_l[n]=\sum_{q=1}^{Q}\beta_q{x}_l[n-n_q]\mathrm{e}^{\mathrm{j} 2\mathrm{\uppi}f_{\mathrm{d},q}lT}+{z}[n],\, 0\leq n-lN<N, 
\end{align}
where $n=t/T_{\mathrm{s}}$, $n_q=\tau_q/T_{\mathrm{s}}$, $N=T/T_{\mathrm{s}}$, and ${x}_l[n]$ is the discrete-time form of the transmitted baseband signal in (\ref{eq:sensingsignalx1}).

\subsubsection{CP removal}
In accordance with the standardized OFDM reception procedure, the first $N_\mathrm{CP}$ sampling points corresponding to the CP shall be removed, where $N_\mathrm{CP}=T_\mathrm{CP}/T_{\mathrm{s}}$. After CP removal, the received signal of the $l$-th symbol from the $q$-th path is expressed as
\begin{align}\label{eq:xlq1}
	{\mathbf{y}}_{l,q}=\beta_q{\mathbf{x}}_{l,q}\mathrm{e}^{\mathrm{j} 2\mathrm{\uppi}f_{\mathrm{d},q}lT}\in\mathbb{C}^{N_\mathrm{D}\times 1},
\end{align}
where $N_\mathrm{D}=T_\mathrm{D}/T_{\mathrm{s}}$ and $N_\mathrm{CP}+N_\mathrm{D}=N$. Under ideal conditions $n_q<N_\mathrm{CP}$, we have
\begin{align}\label{eq:ylq2}
	{\mathbf{x}}_{l,q}=&\left[{s}[N_\mathrm{D}-n_q],...,{s}[N_\mathrm{D}-1],\right.\\&\left.{s}[0],...,{s}[N_\mathrm{D}-1-n_q]\right]^{\mathsf{T}},\nonumber
\end{align}
where ${s}[n]$ is the discrete-time form of the OFDM symbol in (\ref{eq:sensingsignals1}).
\begin{remark}
	The signal ${\mathbf{x}}_{l,q}$ is a cyclically shifted version of the transmitted pilot ${\mathbf{s}}=[{s}[0],...,{s}[N_\mathrm{D}-1]]^{\mathsf{T}}$, thereby leaving the structure of the frequency-domain pilot intact.
\end{remark}

\subsubsection{Pilot unloading}
 To extract the target sensing information from the channel response, it is necessary to perform pilot unloading analogous to pilot-aided channel estimation. After the Fourier transform of ${\mathbf{y}}_{l}=\sum_{q=1}^{Q}{\mathbf{y}}_{l,q}$, the ideal signal on the $k$-th subcarrier of the $l$-th symbol in the frequency domain is derived as
\begin{align}\label{eq:sensingsignalY1}
	{{Y}}_{k,l}=\sum_{q=1}^{Q}\beta_q S_k\mathrm{e}^{-\mathrm{j} 2\mathrm{\uppi}\tau_{q}k\Delta_\mathrm{f}}\mathrm{e}^{\mathrm{j} 2\mathrm{\uppi}f_{\mathrm{d},q}lT}+{{Z}}_{k,l}.
\end{align}
Then, we perform pilot unloading to obtain channel state information (CSI) as
\begin{align}\label{eq:sensingsignalH1}
	{{H}}_{k,l}={{Y}}_{k,l}/{{S}}_k=\sum_{q=1}^{Q}\beta_q\mathrm{e}^{-\mathrm{j} 2\mathrm{\uppi}\tau_{q}k\Delta_{\mathrm{f}}}\mathrm{e}^{\mathrm{j} 2\mathrm{\uppi}f_{\mathrm{d},q}lT}+\bar{{Z}}_{k,l},
\end{align}
where $\bar{{Z}}_{k,l}={{Z}}_{k,l}/{{S}}_{k}$. This process is analogous to least squares channel estimation in communication systems. 
\begin{remark}
	Subsequently, CSI is exploited to obtain the RV spectrum by two-dimensional discrete Fourier transform (2D-DFT) and enable target sensing \cite{Shi2022Device}.
\end{remark}

\subsection{Artifacts Origin Analysis}\label{sec:OriginAnalysis}
In real-world systems, non-ideal impairments introduce multiplicative noise by affecting the amplitude and phase of the received signal in the time domain, which is modeled as
\begin{align}\label{appex:appexl0k}
	\tilde{y}_l[n]={g}_l[n]{y}_l[n],
\end{align}
where ${g}_l[n]$ and ${y}_l[n]$ are the multiplicative noise and the ideal received signal of the $n$-th sampling point at the $l$-th symbol, respectively. The multiplicative noise can induce striped artifacts in the RV spectrum, as shown in Fig. \ref{fig:artifacts1}, which are formulated by the following theorems.

\begin{theorem}[velocity/Doppler domain spreading]
	\label{th:period1}
	When utilizing the OFDM signal for sensing, the time-domain multiplicative noise $\{{g}_l[n]\}_{l=0,..,L-1}$ causes spreading in the velocity spectrum. This spreading leads to an irregular rise of the noise floor across the entire velocity spectrum.
\end{theorem}
\begin{IEEEproof}
	The detailed proof is provided in Appendix \ref{AppendixA}.
\end{IEEEproof}
\begin{theorem}[range/delay domain extension]
	\label{th:period2}
	When utilizing the SRS with the ZC sequence for sensing, the time-domain multiplicative noise $\{{g}_l[n]\}_{n=0,..,N-1}$ introduces extensions of the range spectrum. These extensions exhibit the following two characteristics. First, they extend bidirectionally from the center position $c_0\tau_{q}/2$ corresponding to the actual range peak, where $c_0$ is the speed of light. Second, the temporal positions of these extensions demonstrate periodicity with a period of ${c_0\mu}/{2K\Delta_{\mathrm{f}}}$.
\end{theorem}
\begin{IEEEproof}
	The detailed proof is provided in Appendix \ref{AppendixB}.
\end{IEEEproof}

These theorems reveal the velocity-domain spreading and range-domain extension induced by multiplicative noise, respectively.  {We also derive the impact of multiplicative noise on cyclic matrix-polyphase codes in Appendix~\ref{AppendixC}.}  

{The theoretical analysis applies when the non-ideal impairment can be represented as a time-domain pointwise complex multiplication. During one coherent processing interval, the delay and Doppler of each propagation path are assumed to be approximately constant. The pilot phase structure is unchanged across symbols. Under these conditions, \textit{Theorems}~\ref{th:period1} and~\ref{th:period2} characterize the Doppler-domain spreading and the code-dependent range-domain extension, respectively. Their coupling produces the observed RV-domain artifacts.}

{
\begin{remark}
	 Under multi-target, multipath, or strong-reflector propagation, the striped artifacts from individual paths superpose linearly. The resulting artifacts, typically dominated by strong static paths, appear across the entire RV spectrum and consequently obscure weak targets.
\end{remark}
}

Here, we primarily focus on the following two types of non-ideal impairments as examples: time-domain truncation and carrier frequency offset (CFO).
\subsubsection{Analysis of time-domain truncation}

In communication systems, the CP-OFDM waveform can effectively mitigate inter-symbol interference (ISI) and ICI caused by multipath effects. However, when performing long-range target sensing, the current CP length is entirely insufficient to meet the requirements for target sensing at ranges of hundreds of meters or even kilometers.\footnote[2]{In 5G NR specifications, the CP length for a waveform with $\Delta_\mathrm{f}=120\,$kHz is $T_\mathrm{CP}=0.59$ $\mu$s, which supports a maximum unambiguous range of only $88.5$ m for round-trip propagation in target sensing.} 
As shown in Fig. \ref{fig:InCPCFO}(a), even after CP removal, the residual component of the $(l-1)$-th OFDM symbol leaking into the $l$-th symbol cannot be completely eliminated. 
Therefore, ISI and ICI are introduced due to the insufficient CP length.

\begin{figure*}[]
	\centerline{\includegraphics[width=1.1\linewidth]{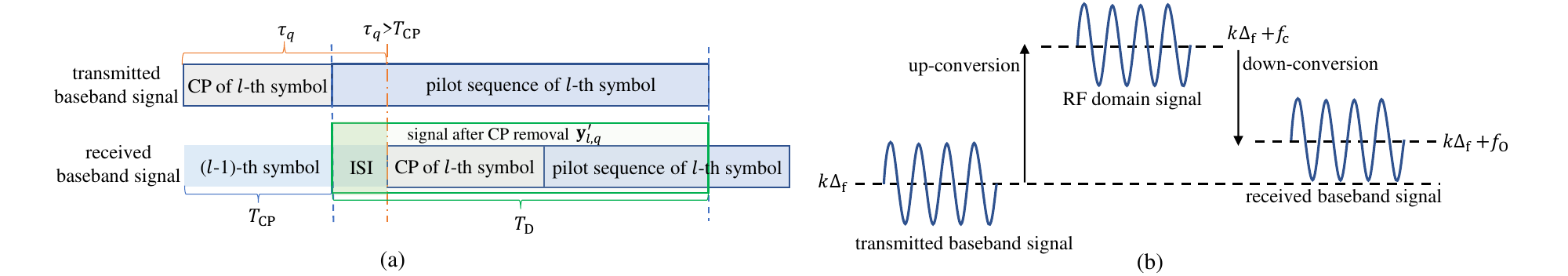}}
	\caption{\textmd{Illustration of the received baseband signal with non-ideal impairments. (a) time-domain truncation. (b) carrier frequency offset.}}\label{fig:InCPCFO}
	\vspace{-1em}
\end{figure*}

Subsequently, we analyze the impact of insufficient CP length on target sensing. Recalling the CP removal process in (\ref{eq:xlq1})--(\ref{eq:ylq2}),
when the delay of $q$-th path exceeds the CP length, i.e. $N_\mathrm{CP}<n_q<N$, the received signal of the $l$-th symbol from the $q$-th path is rewritten as
\begin{align}\label{eq:tlqp}
	\tilde{\mathbf{y}}^{\mathrm{\prime}}_{l,q}=\beta_q\left({\mathbf{x}}^{\mathrm{\prime}}_{l-1,q}+{\mathbf{x}}^{\mathrm{\prime}}_{l,q}\right)\mathrm{e}^{\mathrm{j} 2\mathrm{\uppi}f_{\mathrm{d},q}lT}\in\mathbb{C}^{N_\mathrm{D}\times 1},
\end{align}
where 
\begin{align}\label{eq:ylq31}
	{\mathbf{x}}^{\mathrm{\prime}}_{l-1,q}=\left[{x}_{l-1}[N-\Delta n],...,{x}_{l-1}[N-1],0,...,0\right]^{\mathsf{T}},
\end{align}
\begin{align}\label{eq:ylq32}
	{\mathbf{x}}^{\mathrm{\prime}}_{l,q}=\left[0,...,0,{x}_l[0],...,{x}_l[N-1-n_q]\right]^{\mathsf{T}},
\end{align}
and $\Delta n=n_q-N_\mathrm{CP}>0$. 
In this case, the ISI $\mathbf{x}_{l-1,q}^{\prime}$ from $(l-1)$-th symbol is introduced into the 
$l$-th symbol. 

Meanwhile, the $l$-th symbol $\mathbf{x}_{l,q}^{\prime}$ is truncated in the time domain and consequently ceases to be a cyclic shift. As a result, its frequency-domain structure is disrupted by ICI, and (\ref{eq:sensingsignalY1}) is no longer satisfied. 
Subsequently, we analyze the effects of time-domain truncation on the frequency domain. 
According to (\ref{eq:ylq32}), the received signal with time-domain truncation is expressed as
\begin{align}\label{eq:ylq1111}
	\tilde{\mathbf{y}}_{l,q}^{\mathrm{\prime(Trunc)}}=\beta_q{\mathbf{x}}^{\mathrm{\prime}}_{l,q}\mathrm{e}^{\mathrm{j} 2\mathrm{\uppi}f_{\mathrm{d},q}lT}\in\mathbb{C}^{N_\mathrm{D}\times 1},
\end{align}
where the non-ideal pilot $\mathbf{x}^{\mathrm{\prime}}_{l,q}$ is given as
\begin{align}\label{eq:xpp}
	{x}^{\mathrm{\prime}}_{l,q}[n]={g}^{\mathrm{\prime}}[n]{x}_{l,q}[n],\quad n=0,...,N_\mathrm{D}-1.
\end{align}
and
\begin{align}\label{eq:xl02} {g}^{\mathrm{\prime}}[n]=
	\begin{cases}
		0, & \quad 0\leq n<\Delta n; \\
		1, & \quad \Delta n\leq n < N_\mathrm{D}.\\
	\end{cases}
\end{align}
Following a similar derivation in Appendix \ref{AppendixB}, it can be concluded that time-domain truncation induces periodic extensions in the range domain. Therefore, the extensions $\boldsymbol{\xi}^{\mathrm{\prime}}_{m,l,q}$ caused by the $m$-th ICI ($m>0$) from the $q$-th path are modeled as
\begin{align}\label{eq:Xii}
&\xi^{\mathrm{\prime}}_{m,l,q}[d]=\frac{\beta_q\mathrm{e}^{\jmath 2\mathrm{\uppi}f_{\mathrm{d},q}lT}}{N_\mathrm{D}}
\\
&\times\left({G}_m^{\mathrm{\prime}}\mathrm{e}^{\mathrm{j} 2\mathrm{\uppi}\frac{\mu m(m-1)/2}{K}}\sum_{k=m}^{K-1}\mathrm{e}^{\mathrm{j} 2\mathrm{\uppi}(\tau[d]-\tau_{q}-\delta_\tau[m])k\Delta_{\mathrm{f}}}\right.\nonumber\\&+\left.{G}^{\mathrm{\prime}}_{-m}\mathrm{e}^{\mathrm{j} 2\mathrm{\uppi}\frac{\mu m(m+1)/2}{K}}\sum_{k=0}^{K-m-1}\mathrm{e}^{\mathrm{j} 2\mathrm{\uppi}(\tau[d]-\tau_{q}+\delta_\tau[m])k\Delta_{\mathrm{f}}}\right),\nonumber
\end{align}
where $\tau[d]$ is the delay and $d$ is the corresponding index in the delay domain, and
\begin{align}\label{eq:deltad}
\delta_\tau[m]=\frac{\mu m}{K\Delta_{\mathrm{f}}},\,m>0
\end{align}
and
\begin{align}
	&{G}^{\mathrm{\prime}}_k=\mathcal{F}\{{g}^{\mathrm{\prime}}[n]\}=\sum_{n=0}^{N_\mathrm{D}-1}{g}^{\mathrm{\prime}}[n]\mathrm{e}^{-\mathrm{j} 2\uppi nk/ N_\mathrm{D}}\\
	&=
	\begin{cases}
		N_\mathrm{D}-\Delta n & ,\quad k=0,\nonumber \\
		(-1)^k\mathrm{e}^{-\mathrm{j} \uppi k (\Delta n-1)/ N_\mathrm{D}}\frac{\sin({\uppi k (N_\mathrm{D}-\Delta n)/ N_\mathrm{D}})}{\sin(\uppi k/ N_\mathrm{D})} & .\quad k\neq0.\nonumber
	\end{cases}
\end{align}
At this point, the range-domain extension caused by time-domain truncation has been derived. When coupled with velocity-domain spreading, it appears as striped artifacts in the RV spectrum.


\subsubsection{Analysis of CFO}
In practical systems, the instability of the local oscillators leads to carrier frequency mismatch between the transmitter and receiver, resulting in CFO as shown in Fig. \ref{fig:InCPCFO}(b).\footnote[3]{In a mmWave system with a carrier frequency of $f_{\mathrm{c}}=25.6$ GHz, the CFO $f_{\mathrm{o}}$ typically remains below $5$ kHz. Compared to communication, the sensing performance suffers more severe degradation.}

Recalling the signal under ideal assumptions in (\ref{eq:xlq1}), the received pilot signal with CFO from the $q$-th path is expressed as
\begin{align}\label{eq:ylq111}
\tilde{\mathbf{y}}_{l,q}^{\mathrm{\prime\prime}}=\beta_q{\mathbf{x}}^{\mathrm{\prime\prime}}_{l,q}\mathrm{e}^{\mathrm{j} 2\mathrm{\uppi}f_{\mathrm{d},q}lT}\in\mathbb{C}^{N_\mathrm{D}\times 1},
\end{align}
where the non-ideal pilot ${\mathbf{x}}^{\mathrm{\prime\prime}}_{l,q}$ is given as
\begin{align}\label{eq:xlq111}
{{x}}^{\mathrm{\prime\prime}}_{l,q}[n]=g^{\mathrm{\prime\prime}}[n]{{x}}_{l,q}[n],\quad n=0,...,N_\mathrm{D}-1
\end{align}
and 
\begin{align}
g^{\mathrm{\prime\prime}}[n]=\mathrm{e}^{\mathrm{j} 2\mathrm{\uppi}(n+N_\mathrm{CP})f_{\mathrm{o}}T_{\mathrm{s}}},
\end{align}
where $f_{\mathrm{o}}$ is the CFO. 

Following a similar derivation in Appendix \ref{AppendixB}, it can be concluded that CFO induces periodic extensions in the range domain. The extensions $\boldsymbol{\xi}^{\mathrm{\prime\prime}}_{m,l,q}$ caused by the $m$-th ICI from the $q$-th path are modeled similarly to (\ref{eq:Xii})--(\ref{eq:deltad}), and we have
\begin{align}\label{eq:gprim}
	&{G}^{\mathrm{\prime\prime}}_k=\mathcal{F}\{g^{\mathrm{\prime\prime}}[n]\}=\sum_{n=0}^{N_\mathrm{D}-1}{g}^{\mathrm{\prime\prime}}[n]\mathrm{e}^{-\mathrm{j} 2\uppi nk/ N_\mathrm{D}}\\&=\mathrm{e}^{\mathrm{j} \mathrm{\uppi}(2N_\mathrm{CP}+N_\mathrm{D}-1)f_{\mathrm{o}}T_{\mathrm{s}}}\mathrm{e}^{\mathrm{j} \uppi k/ N_\mathrm{D}}\frac{\sin({\uppi f_{\mathrm{o}}T_{\mathrm{s}} N_\mathrm{D}})}{\sin(\uppi (f_{\mathrm{o}}T_{\mathrm{s}}N_\mathrm{D}-k)/ N_\mathrm{D})}.\nonumber
\end{align}
It can be observed that the extension originating from CFO is consistent with that resulting from time-domain truncation. When coupled with velocity-domain spreading, it also appears as striped artifacts in the RV spectrum.
\begin{remark}
Note that when $f_{\mathrm{o}}=i/(T_{\mathrm{s}} N_\mathrm{D})=i\Delta_{\mathrm{f}}$, (\ref{eq:gprim}) will no longer hold, where $i$ denotes an integer. Under this condition, we have
\begin{align}
	{G}^{\mathrm{\prime\prime}}_k=
	\begin{cases}
		N_{\mathrm{D}}\mathrm{e}^{\mathrm{j} 2\uppi i N_\mathrm{CP}/ N_\mathrm{D}} & ,\quad k\equiv i\,(\text{mod}\, N_\mathrm{D}),\nonumber \\
		0 & ,\quad \text{others}.\nonumber
	\end{cases}
\end{align} 
This demonstrates that when the CFO is an integer multiple of the subcarrier spacing, the range domain does not exhibit extension, but only a shift in the spectrum peak.
\end{remark}



\subsection{Problem Formulation}
The foregoing analysis has identified the origin of the artifacts, and it is conclusive that such interference inevitably impairs target sensing. As shown in Fig. \ref{fig:artifacts1}, these artifacts result in the masking of weak targets and false alarms in the RV spectrum. Let ${{X}}[d,v]$ be the amplitude of the RV spectrum, which is expressed as
\begin{align}\label{eq:RV0}
	{{X}}[d,v]=\alpha[d,v]{{X}}^{\rm{(0)}}[d,v]+\xi[d,v],
\end{align}
where $\boldsymbol{{{X}}}^{\rm{(0)}}$, $\boldsymbol{\alpha}$, and $\boldsymbol{\xi}$ are the ideal RV spectrum, multiplicative interference, and additive interference, respectively. $d$ and $v$ are the indices of range and velocity, respectively. Therefore, our task is to reconstruct $\boldsymbol{{{X}}}^{\rm{(0)}}$ from the observed $\boldsymbol{{{X}}}$, thereby mitigating the impact of artifacts caused by interference.
However, these impairments are often inconspicuous and time-varying, posing significant challenges to accurate estimation and compensation for $\boldsymbol{\alpha}$ and $\boldsymbol{\xi}$ by conventional approaches. 

It is worth noting that we have converted the sensing signal from the time-frequency domain into the RV domain. This transformation has significant implications. Established computer vision techniques, such as CNNs, can be utilized for effective denoising and reconstruction, capitalizing on their inherent strengths in processing image-like data. Furthermore, the derived origin and characteristics of the artifacts inform the design of a data-physics-driven artifact elimination method that leverages the physical properties inherent in the RV spectrum.
\begin{figure*}[b]
	\vspace{-0.1cm}
	\centerline{\includegraphics[width=0.87 \linewidth]{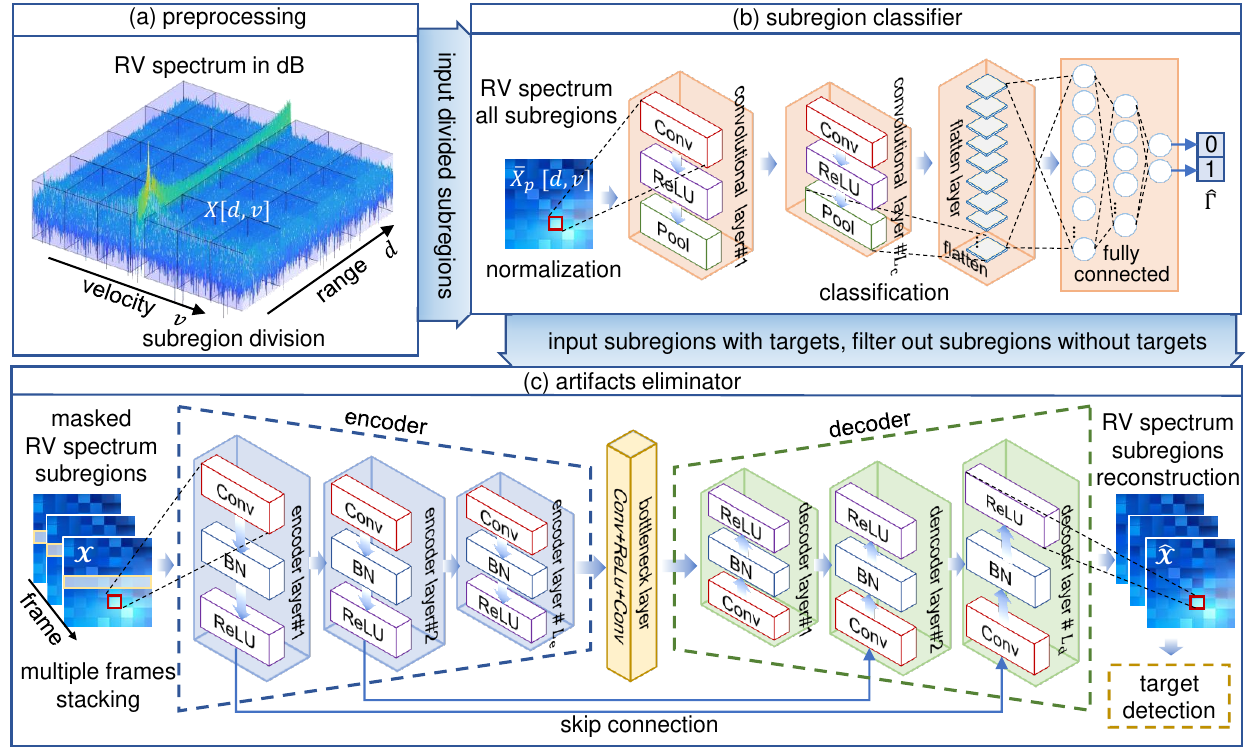}}
	\caption{PIAENet for target detection in the RV spectrum with real-world artifacts. (a) preprocessing module. (b) subregion classifier. (c) artifact eliminator.}\label{fig:frameworkl}
\end{figure*}
\section{Artifact Elimination Methodology}\label{sec:Elimination}
Conventional approaches usually struggle to address the challenges posed by a multitude of real-world phase errors that are time-varying and inconspicuous. In this article, we propose PIAENet, a new data-physics-driven framework for target detection while accounting for the presence of artifacts. The proposed framework, as visualized in Fig. \ref{fig:frameworkl}, comprises three cascade-connected modules: a preprocessing module, a subregion classifier, and an artifact eliminator.

In the subsequent subsections, we provide a comprehensive description of each module within the proposed framework.
\subsection{Preprocessing Module}
In sensing systems, the direct path between the transmitter and receiver or strong multipath components in the environment often exhibits extremely high amplitudes. 
The amplitude of the strong multipath component is more than $60$ dB higher than that of other paths, as shown in Fig. \ref{fig:artifacts1}. Therefore, if the entire RV spectrum is used as input, the neural network can extract only the features of the strong multipath components, while failing to preserve the features of weak targets. 

To enhance the sensitivity to weak targets, we divide the RV spectrum into multiple subregions. The $p$-th subregion is described as
\begin{align}\label{eq:RVsub}
	{{X}}_p[d,v]={{X}}[d,v],\, d\in \mathbf{d}_p,\, v\in \mathbf{v}_p, 
\end{align}
where $\mathbf{d}_p\in\mathbb{R}^{K'\times 1}$ and $\mathbf{v}_p\in\mathbb{R}^{L'\times 1}$ are the range and velocity bins of $p$-th subregion. These strong multipath components tend to appear around the bins of zero-range or zero-velocity. To ensure coverage of all bins except those at zero-velocity or zero-range, $\mathbf{d}_p$ and $\mathbf{v}_p$ must satisfy
\begin{subequations}
	\begin{align}\label{eq:RVcub}
		&\bigcup_{p=1}^P\mathbf{d}_p=[d_{\mathrm{min}},...,d_{\mathrm{max}}]^{\mathsf{T}}, \\ &\bigcup_{p=1}^P\mathbf{v}_p=[-v_{\mathrm{max}},...,-v_{\mathrm{min}},v_{\mathrm{min}},...,v_{\mathrm{max}}]^{\mathsf{T}},
	\end{align}
\end{subequations}
where $d_{\mathrm{min}}>\frac{c_0}{2K\Delta_{\mathrm{f}}}$, $d_{\mathrm{max}}\leq\frac{c_0}{2\Delta_{\mathrm{f}}}$, $v_{\mathrm{min}}>\frac{\lambda}{4LT}$ and $v_{\mathrm{max}}\leq\frac{\lambda}{4T}$. 
{The moving target indication (MTI) algorithm can also suppress part of the static strong-path interference. However, its energy spread into nonzero-Doppler bins still degrades sensing performance.}

Furthermore, the amplitude of each RV bin in an individual subregion exhibits a large dynamic range. 
%
%
Therefore, we perform normalization to confine the amplitude within the range of $0$ to $1$. For the $p$-th subregion $\boldsymbol{{{X}}}_p\in\mathbb{R}^{K'\times L'}$, the normalization process is expressed as
\begin{align}\label{eq:RVnor}
	\bar{{{X}}}_p[d,v]={{X}}_p[d,v]/a_p,\,
\end{align}
where $a_p=\max(\boldsymbol{{{X}}}_p)$. By controlling the amplitude range of the network input, we prevent most neuronal inputs from falling into saturation regions, thereby ensuring sufficient gradients during backpropagation for effective weight updates.

\begin{figure*}[]
	\centerline{\includegraphics[width=0.88 \linewidth]{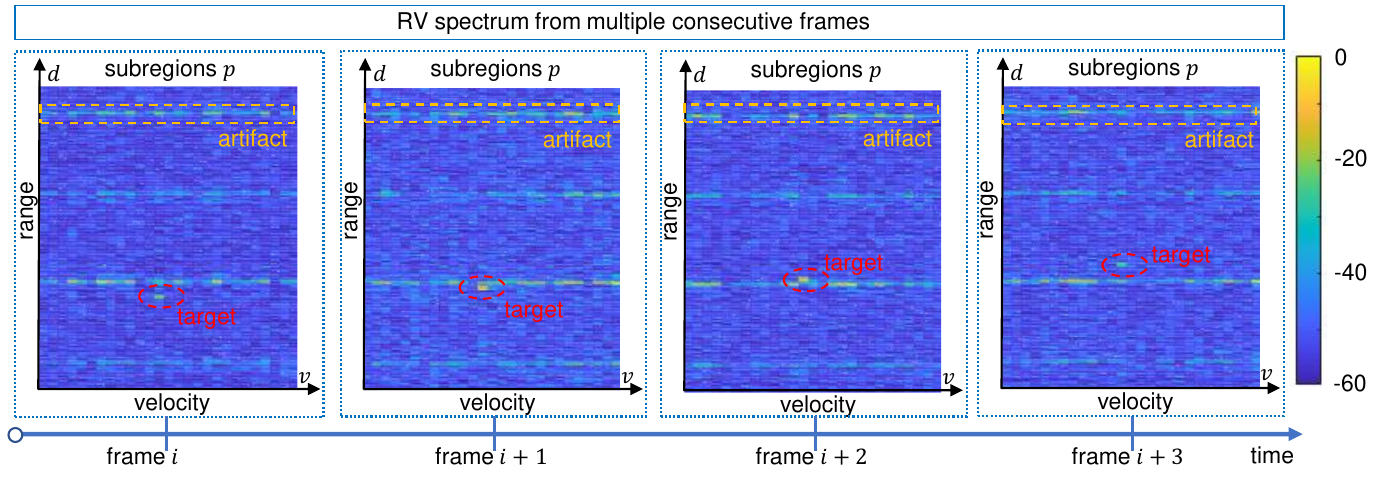}}
	\caption{Raw spectrum of subregion $p$ from multiple consecutive frames. The red ellipse indicates the target, while the orange rectangle marks the artifact.}\label{fig:subregion}
	\vspace{-0.5cm}
\end{figure*}
\vspace{-1.5em}
\subsection{Subregion Classifier}
After subregion division and normalization, two types of samples are formed: subregions with targets and subregions without targets. To optimize the training set, we filter out subregions without targets, as they contribute only noisy components and provide no benefit to subsequent target detection. This problem can be formulated as a standard image classification task, and binary image classification is a well-established problem in the field of computer vision. Therefore, we train a CNN-based classifier \cite{CNNclass} to select subregions with targets. For a sample $\bar{\boldsymbol{{{X}}}}_p\in\mathbb{R}^{1\times K'\times L'}$, the corresponding output of the CNN is expressed as
\begin{align}\label{eq:classify1}
	\eta_p=\mathbf{W}^{\mathsf{T}}f_{\mathrm{CNN}}\left(\bar{\boldsymbol{{{X}}}}_p;\boldsymbol{\theta}_{\mathrm{C}}\right)+b,
\end{align}
where $\mathbf{W}$ is the weight vector of the fully connected layer, $f_{\mathrm{CNN}}(\cdot)$ is the output vector of the $L_{\mathrm{c}}$-layer CNN for feature extraction, $\boldsymbol{\theta}_{\mathrm{C}}$ is the network parameter, and $b$ is the bias term. 

Then, the output of the fully connected layer is converted into probability values using an activation function. Employing the Sigmoid function as the activation function, the probability for binary classification is described as
\begin{subequations}
	\begin{align}\label{eq:classify2}
		&P(\Gamma=1|\bar{\boldsymbol{{{X}}}}_p)=\sigma\left(\eta_p\right)=\frac{1}{1+\mathrm{e}^{-\eta}}, \\ &P(\Gamma=0|\bar{\boldsymbol{{{X}}}}_p)=1-\sigma\left(\eta_p\right)=\frac{1}{1+\mathrm{e}^{\eta}},
	\end{align}
\end{subequations}
where $\sigma(\cdot)$ is the Sigmoid function, and $\Gamma=1$ and $\Gamma=0$ represent subregions with and without targets, respectively. Thus, we can obtain the final classification result as
\begin{align}\label{eq:classify3}
	\hat{\Gamma}=
	\begin{cases}
		1, & P(\Gamma=1|\bar{\boldsymbol{{{X}}}}_p)\geq0.5; \\
		0, & \mathrm{otherwise}.\\
	\end{cases}
\end{align}
During training, binary cross-entropy is employed as the loss function, which is expressed as
\begin{align}\label{eq:classify4}
	\mathcal{L}_{\mathrm{c}}=-\left(\Gamma\log\left(\sigma\left(\eta_p\right)\right)+(1-\Gamma)\log\left(1-\sigma\left(\eta_p\right)\right)\right),
\end{align}
 where $\Gamma\in\{0,1\}$ is the label. The CNN-based classifier acts as a pre-filtering module to remove subregions without targets. This ensures that the downstream training process is concentrated exclusively on target features, thereby improving learning efficiency.

\subsection{Artifacts Eliminator: Physical Priors}
These samples with targets identified through classification are input into the artifact eliminator to enhance target detection performance. To handle the image-like RV spectrum, an encoder-decoder \cite{UNet} can be implemented to effectively mitigate artifacts arising from non-ideal impairments. When targets are not obscured by artifacts, the network can effectively capture their features and adequately suppress artifact interference. However, when weak targets are obscured by artifacts, their features in the image-like spectrum cannot be effectively extracted by the network.

To address this challenge, we propose a multi-frame selective-masked encoder-decoder by incorporating physical priors into the training process. The introduced physical priors pertain to the temporal and spatial characteristics exhibited by the targets and artifacts in the RV spectrum. On the one hand, the position of a target in the RV spectrum is expected to vary continuously between consecutive time frames. On the other hand, these artifacts appear at deterministic locations in the RV spectrum. As dictated by motion dynamics \cite{Khalid2019Convolutional}, target motion follows temporal continuity, meaning its position does not change abruptly and exhibits a strong dependence between consecutive frames. Recalling the artifact periodicity from \textit{Theorem} \ref{th:period2}, this spatial consistency stems from the physical origin of the artifacts and is deterministically tied to the pilot sequences used, thus not varying randomly. To illustrate these characteristics, we present a set of RV spectra for a specific region from multiple consecutive frames in Fig. \ref{fig:subregion}. 

The proposed encoder-decoder leverages temporal continuity and spatial consistency to reconstruct the RV spectrum and thereby enhance target detection performance. First, to prevent weak targets from being obscured by artifacts, we employ a selective mask to guide the network to focus on artifact-affected regions. According to the periodicity of artifacts, the masked range bin is expressed as
\begin{align}\label{eq:SM}
	\mathbf{d}_{\mathrm{M}}=\left\{d\in\mathbf{d} | d = d_0 \pm k\frac{c_0\mu}{2K\Delta_{\mathrm{f}}},\,k=1,\cdots,K_{\mathrm{M}} \right\},
\end{align}
where $\mathbf{d}=[0,...,\frac{c_0}{2\Delta_{\mathrm{f}}}]^{\mathsf{T}}$, $K_{\mathrm{M}}=\lceil\frac{K}{\mu}\rceil$ and $d_0=\arg\max\limits_{d} \left\arrowvert{{X}}[d,v]\right\arrowvert$ is the range bin of the strongest path. 

Then, to capture the temporal relationships between consecutive frames, we employ a multi-frame stacking input and incorporate a module that explicitly models temporal dynamics. By stacking samples from multiple frames, we obtain the multi-frame sample ${\boldsymbol{{\mathcal{X}}}}\in\mathbb{R}^{I\times K'\times L'}$ as
\begin{align}\label{eq:ED1}
	{{{\mathcal{X}}}}[i ,d,v]=\bar{{{{X}}}}^{i}_{\mathrm{(1)}}[d,v]\,,i=1\cdots I,
\end{align}
where $I$ is the number of stacked frames, and $\bar{\boldsymbol{{{X}}}}^{i}_{\mathrm{(1)}}$ is the target-containing subregion in the $i$-th frame of the masked RV spectrum. 

\begin{remark}
When stacking samples from different frames, the subregion $p$ should remain consistent across frames, even if some frames are classified as subregions without targets by the classifier. For example, in consecutive frames $i-1$, $i$, and $i+1$, if subregion $p$ contains a target in both frame $i-1$ and frame $i+1$, all three frames should be stacked together, even though the classifier categorizes subregion $p$ in frame $i$ as a subregion without targets.
\end{remark}

\subsection{Artifacts Eliminator: Training Mechanism}
Consequently, the three-dimensional tensor $\boldsymbol{\mathcal{X}}$ serves as the input to the proposed encoder-decoder. The number of consecutive frames is utilized as the input channel dimension, implicitly capturing temporal correlations. As shown in Fig. \ref{fig:frameworkl}(c), the artifact eliminator is regarded as a reconstruction process of the RV spectrum, which is expressed as
\begin{align}\label{eq:ED0}
	\hat{\boldsymbol{{\mathcal{X}}}}={D}\left({E}\left({\boldsymbol{{{\mathcal{X}}}}};\boldsymbol{\theta}_{\mathrm{E}}\right);\boldsymbol{\theta}_{\mathrm{D}}\right)\,
\end{align}
where ${E}(\cdot)$, ${D}(\cdot)$, $\boldsymbol{\theta}_{\mathrm{E}}$, and $\boldsymbol{\theta}_{\mathrm{D}}$ are the encoder function, decoder function, encoder network parameters, and decoder network parameters, respectively. 
Let the number of layers of the encoder be $L_{\mathrm{e}}$, and the $l_{\mathrm{e}}$-th layer of the encoder is described as
\begin{align}\label{eq:El}
	\boldsymbol{\mathcal{E}}_{l_{\mathrm{e}}}=f_{\mathrm{E}}^{l_{\mathrm{e}}}(\boldsymbol{\mathcal{E}}_{l_{\mathrm{e}}-1};\boldsymbol{\theta}_{l_{\mathrm{e}}})\,
\end{align}
where $f_{\mathrm{E}}^{l_{\mathrm{e}}}(\cdot)$ is the $l_{\mathrm{e}}$-th-layer encoder function, which can be viewed as a downsampling process. $\boldsymbol{\theta}_{l_{\mathrm{e}}}$ is the network parameter of the $l_{\mathrm{e}}$-th encoder layer. For the first encoder layer, its input is $\boldsymbol{\mathcal{E}}_{0}=\boldsymbol{{{\mathcal{X}}}}$. Between the encoder and the decoder is the bottleneck layer, which is expressed as
\begin{align}\label{eq:B}
	\boldsymbol{\mathcal{B}}=f_{\mathrm{B}}(\boldsymbol{\mathcal{E}}_{L_{\mathrm{e}}};\boldsymbol{\theta}_{{\mathrm{B}}})\,
\end{align}
where $f_{\mathrm{B}}(\cdot)$ is the function of the bottleneck layer, and $\boldsymbol{\theta}_{{\mathrm{B}}}$ is the corresponding network parameter. Here, we adopt a structure that first increases and then decreases the dimension to keep the spatial size unchanged and preserve key spatial information. This structure is implemented by two convolutional layers with an activation function inserted in between. Subsequently, we adopt a symmetric structure between the decoder and the encoder, and the number of layers of the decoder is set as $L_{\mathrm{d}}=L_{\mathrm{e}}$ to achieve symmetric feature extraction and restoration. 
To mitigate the loss of spatial information during downsampling, skip connections are incorporated between the encoder and decoder. Thus, the decoder with skip connection is described as
\begin{align}\label{eq:Dl}
	\boldsymbol{\mathcal{D}}_{l_{\mathrm{d}}}=f_{\mathrm{D}}^{l_{\mathrm{d}}}(\boldsymbol{\mathcal{D}}_{l_{\mathrm{d}}-1},\boldsymbol{\mathcal{E}}_{L_{\mathrm{e}}-l_{\mathrm{d}}+1};\boldsymbol{\theta}_{l_{\mathrm{d}}})\,
\end{align}
where $f_{\mathrm{D}}^{l_{\mathrm{d}}}(\cdot)$ is the $l_{\mathrm{d}}$-th-layer decoder function, which can be viewed as an upsampling process. $\boldsymbol{\theta}_{l_{\mathrm{d}}}$ is the network parameter of the $l_{\mathrm{d}}$-th decoder layer. For the first decoder layer, its input is $\boldsymbol{\mathcal{D}}_{0}=\boldsymbol{{{\mathcal{B}}}}$. The skip connection allows the decoder to utilize both compressed feature representations and fine-grained spatial details. We employ a supervised training approach using the mean squared error (MSE) loss function, which is defined as
\begin{align}\label{eq:elimloss}
	\mathcal{L}_{\mathrm{ED}}=\frac{1}{I K' L'}||\hat{\boldsymbol{{\mathcal{X}}}}-{\boldsymbol{{\mathcal{X}}}}^{(0)}||^2_2,
\end{align}
where $\hat{\boldsymbol{{\mathcal{X}}}}=\boldsymbol{\mathcal{D}}_{L_{\mathrm{d}}}$ is the reconstructed RV spectrum and $ {\boldsymbol{{\mathcal{X}}}}^{(0)}$ is the ideal RV spectrum, which can be generated according to the target's range and velocity. 

\begin{remark}
	To ensure the response speed and continuity of detection in real-world deployment, we adopt an overlapping sliding window strategy to build the training and test sets. This involves sampling consecutive frames with a fixed-length window that slides forward by one frame each time.
\end{remark}

\section{Experiment}\label{sec:Experimental}
We conducted a series of data collection campaigns and experiments to validate artifact characteristics and evaluate the proposed framework. First, the mmWave active antenna units (AAUs) were deployed on a rooftop to collect background environmental data. These data were used to observe artifact characteristics and validate the theoretical analysis. Subsequently, data collection with unmanned aerial vehicle (UAV) targets was conducted by deploying the AAUs in a test field. These data were used for performance validation of artifact elimination. {We further introduce simulated data to validate the robustness of the proposed framework under a broader range of system parameters and target trajectories.}

\subsection{Experimental Setup and Datasets}
The experimental details for different scenarios are presented below. The system parameters and dataset settings are listed in Table \ref{table:setting1} and Table \ref{table:setting2}.

\begin{figure}[b]
	\centerline{\includegraphics[width=1.05\linewidth]{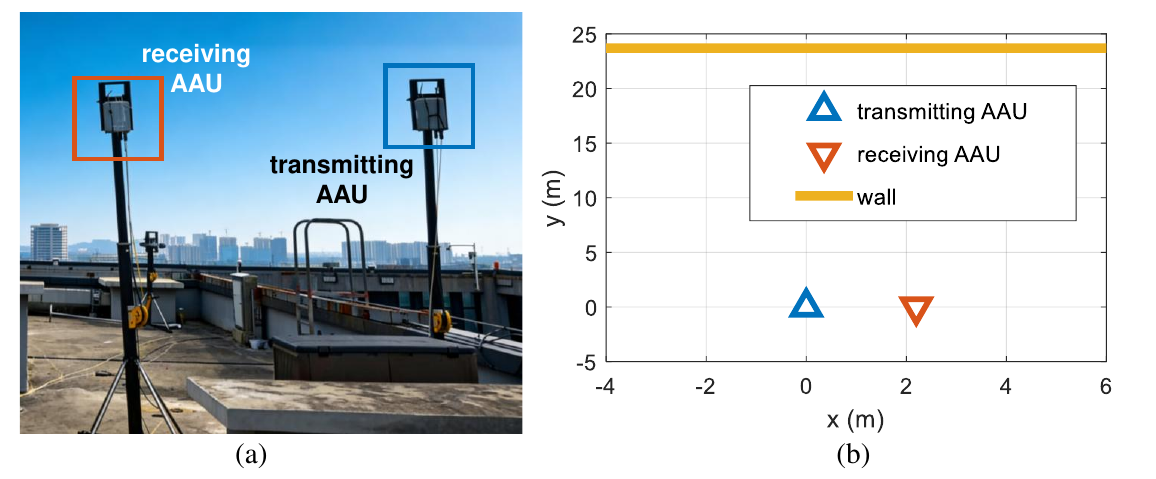}}
	\vspace{-1.1em}
	\caption{\textmd{Experimental setup for {scenario A}. (a) Experimental environment and hardware setup on a rooftop. (b) positions of the transmitting and receiving AAUs.}}\label{fig:AAUposition}
	\vspace{-1.25em}
\end{figure}

\begin{figure}[b]
	\centerline{\includegraphics[width=1.05\linewidth]{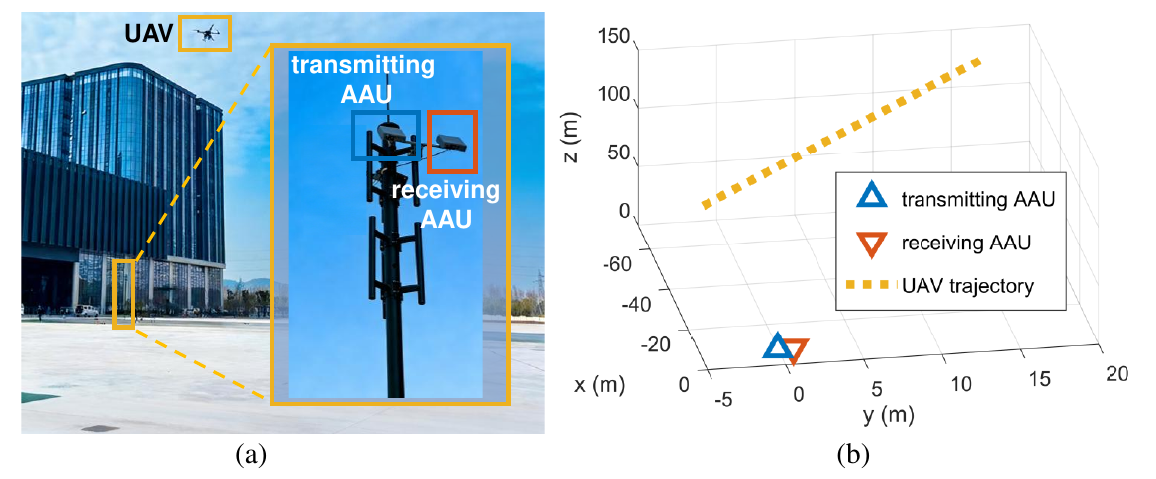}}
	\vspace{-1.1em}
	\caption{\textmd{Experimental setup for {scenario B}. (a) Experimental environment and hardware setup in a test field. (b) positions of the transmitting and receiving AAUs and trajectory of the UAV.}}\label{fig:AAUposition1}
	\vspace{-1.5em}
\end{figure}

\begin{table}[]
	\vspace{-1.75em}
	\centering
	\caption{Parameters of the experimental system.}
	\label{table:setting1}
	\vspace{-0.5em}
	\begin{tabular}{c c c}
		\toprule
		& & \\[-10pt]
		parameter&Scenario A&Scenario B\\
		\hline
		& & \\[-6pt]
		center frequency $f_{\mathrm{c}}$&$25.6\;\mathrm{GHz}$&$25.6\;\mathrm{GHz}$\\
		& & \\[-6pt]
		subcarrier spacing $\Delta_{\mathrm{f}}$&$120\;\text{kHz}$&$120\;\text{kHz}$\\
		& & \\[-6pt]
		number of subcarriers $K$&$1584$&1024\\
		& & \\[-6pt]
		symbol repetition interval $T$&$0.625\;$ms&$0.156\;$ms\\
		& & \\[-6pt]
		number of symbols $L$&$256$&256\\[-1pt]
		\bottomrule
	\end{tabular}
	\vspace{-1.5em}
\end{table}

\begin{table}[]
	\vspace{-0em}
	\centering
	\caption{Dataset settings}
	\label{table:setting2}
	\vspace{-0.5em}
	\begin{tabular}{c c c }
		\toprule
		& & \\[-10pt]
		&Scenario A&Scenario B\\
		\hline
		& & \\[-3.5pt]
		type&\makecell{static environment}&\makecell{with moving UAV}\\
		& & \\[-5pt]
		\makecell{sample size}&$1584\times 256$&$1024\times 256$\\
		& & \\[-5pt]
		\makecell{sample count}&$6$&\makecell{$1216$ (train)\\$304$ (test)}\\
		\bottomrule
	\end{tabular}
	\vspace{-1.5em}
\end{table}

%

\subsubsection{Scenario A}
We collected environmental sensing data using real-world mmWave devices on an outdoor rooftop. As illustrated in Fig. \ref{fig:AAUposition}, the transmitting and receiving AAUs were placed at $(0.0, 0.0, 0.0)$ m and $(2.2, 0.0, 0.0)$ m, respectively. The normal direction of the antenna arrays for both the transmitting and receiving AAUs was aligned parallel to the positive y-axis. A static reflecting wall was positioned 23.7 m away from the AAUs along the positive y-axis. The dataset was collected using three types of ZC sequences, with two samples captured for each sequence configuration, resulting in a total of six samples.

\subsubsection{Scenario B}
We collected sensing data containing UAV target echoes using real-world mmWave devices in a test field.\footnote[4]{The dataset is released for research purposes and is publicly available at: http://pmldatanet.com.cn/dataapp/multimodal.} As shown in Fig. \ref{fig:AAUposition1}, the transmitting and receiving AAUs were installed on pole-mounted brackets centered at $(0.0, 0.0, 13.7)$ m. The UAV target traveled along a linear trajectory from ($-29.0$, $-3.5$, $89.0$) m to ($-72.0$, $17.0$, $123.0$) m at a constant speed of $8$ m/s. The UAV was equipped with a real-time kinematic (RTK) positioning module with centimeter-level accuracy, which provides the ground-truth position and velocity of the UAV. {As shown in Fig.~\ref{fig:RTK}, the ground-truth bistatic range and radial velocity are obtained from RTK measurements using the known Tx/Rx geometry, and are then used to generate the ideal RV spectrum.}

\begin{figure}[h]
	\vspace{-1.1em}
	\centerline{\includegraphics[width=0.85\linewidth]{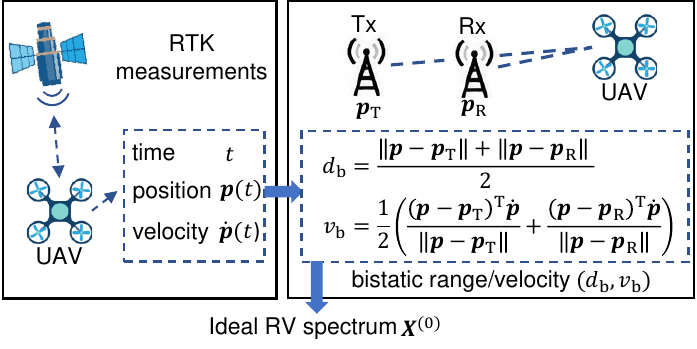}}
	\vspace{-1.1em}
	\caption{\textmd{Generation of the ideal RV spectrum from RTK measurements.}}\label{fig:RTK}
	\vspace{-1.0em}
\end{figure}

\begin{figure*}[]
	\centerline{\includegraphics[width=1\linewidth]{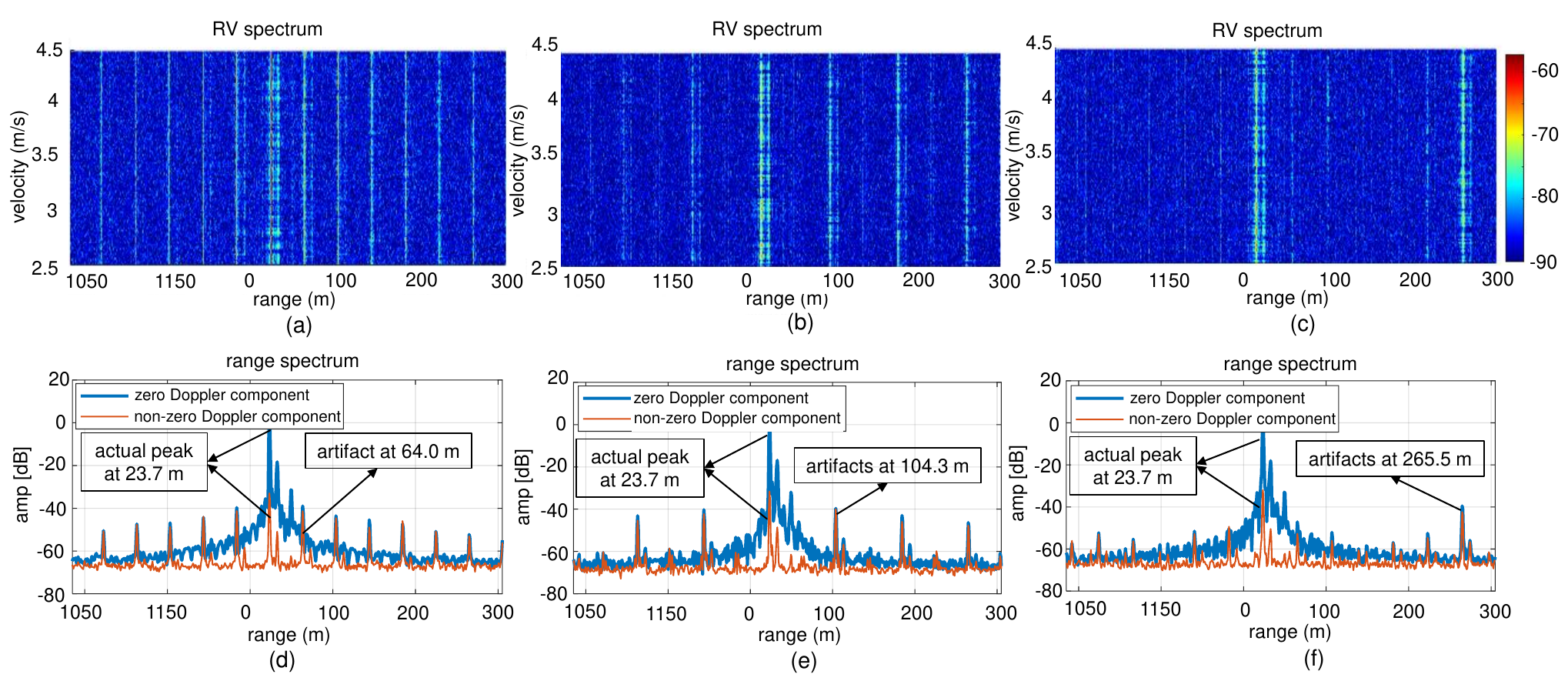}}
	\vspace{-1.5em}
	\caption{\textmd{Experimental results of the normalized RV spectrum and range spectrum with measured data in {scenario A}. $\Delta_{\mathrm{f}}=120\,\mathrm{kHz}$. $K=1584$. (a) RV spectrum. $\mu=51$. (b) RV spectrum. $\mu=102$. (c) RV spectrum. $\mu=306$. (d) range spectrum. $\mu=51$ and $\frac{c_0\mu}{2K\Delta_{\mathrm{f}}}\approx 40.3\,\mathrm{m}$. (e) range spectrum. $\mu=102$ and $\frac{c_0\mu}{2K\Delta_{\mathrm{f}}}\approx 80.5\,\mathrm{m}$. (f) range spectrum. $\mu=306$ and $\frac{c_0\mu}{2K\Delta_{\mathrm{f}}}\approx 241.5\,\mathrm{m}$.}}\label{fig:CFOreal}
	\vspace{-1em}
\end{figure*}

{In scenario B, data were collected over five independent UAV flights, each with two RF channels and eight target-containing beams (out of $30$ scanned), and $23$ frames per flight. A length-$5$ sliding window applied within each flight yields $19$ samples per flight-channel-beam combination, giving a total of $5\times2\times8\times19=1520$ samples. To avoid information leakage from overlapping sliding-window samples, the train/test split is performed at the level of entire flights. The first four flights ($1216$ samples) are used for training and the last flight ($304$ samples) for testing, ensuring no shared raw frames between the two sets. } 

\vspace{-1em}

\subsection{Verification of Artifacts Characteristics}
To verify the theoretical derivations regarding artifact characteristics, the measured data from {scenario A} were first analyzed. In {scenario A}, we employed ZC sequences with different $\mu$ as transmitting pilots and analyzed the received echoes reflected or scattered by objects in the environment. Fig. \ref{fig:CFOreal} shows the RV spectrum and range spectrum of the received echoes. 
The nonzero-Doppler components of the RV spectrum are displayed in Fig. \ref{fig:CFOreal}(a)(b)(c). 
According to \textit{Theorem} \ref{th:period1}, phase noise induces Doppler-domain spectral spreading of the target echo. 
When the energy of the reflected echo is excessively strong, artifacts appear in the nonzero-Doppler components. It can be clearly observed that these artifacts exhibit periodic extensions in the range spectrum, which is consistent with \textit{Theorem} \ref{th:period2}. {The strong static reflector in scenario A is located at $23.7$~m, which is within the CP-supported range of $88.5$~m for $\Delta_\mathrm{f}=120$~kHz. Therefore, the visible artifacts from this reflector are primarily attributable to residual CFO or local-oscillator phase noise rather than insufficient-CP-induced truncation.}

To further validate periodicity, we examined the extension in the range spectrum with different $\mu$. As shown in Fig. \ref{fig:CFOreal}(d)(e)(f), the periodic extension in the zero-Doppler component is dominated by static scatterers, and its periodicity remains discernible. The periodic extension in the nonzero-Doppler components is clearly visible, and its period varies with different $\mu$, which is consistent with \textit{Theorem} \ref{th:period2}. For example, when employing a ZC sequence with $\mu=51$, a false peak extended from the actual peak at $23.7$ m appears at $64.0$ m, which aligns with the calculated period ${c_0\mu}/{2K\Delta_{\mathrm{f}}}\approx 40.3\,\mathrm{m}$. These experimental results validate the characteristics of artifacts in practical OFDM systems with ZC sequences. This provides strong evidence for incorporating physical priors into the network training.

{
The ZC root index also provides a practical waveform-design degree of freedom. The experiments with different root indices suggest two practical guidelines. One is to select $\mu$ so that the artifacts caused by different ICI orders overlap at the same range positions as much as possible. For instance, $\mu=K/2$ and $\mu=K/4$ make the artifacts appear at a limited number of deterministic locations. The other is to select $\mu$ so that the dominant artifact components remain close to the strong reflector. For instance, $\mu=1$ and $\mu=2$ produce small spacings, making distant artifacts associated with higher-order ICI terms and therefore weaker. Nevertheless, the root index must satisfy the applicable SRS configuration, orthogonality, and multiplexing constraints in practical deployments.
}

\subsection{Performance of Artifact Elimination}
We utilized measured data from {scenario B} to evaluate the performance of the proposed architecture in artifact elimination. In {scenario B}, the ZC sequences were used as transmitting pilots, and the echo signal comprises strong scattering paths from the surrounding environment and reflection paths from the UAV. This work focuses on comparing methods for artifact elimination in the RV spectrum, given that the interfering phase errors are time-varying and difficult to calibrate in the time-frequency domain. In addition, our comparison focuses exclusively on supervised learning methods, since accurate labels are available for {scenario B}. {Specifically, we compare the proposed PIAENet against the median filter\footnote[5]{The median filter is adopted to suppress the striped artifacts in the RV spectrum. Specifically, the artifact-prone units identified via calculation based on the periodic characteristics are replaced with the median value of their valid neighboring units.}, the extended cancellation algorithm (ECA), U-Net, and spectral-spatial decomposition for radar signals (Ra-SPD) \cite{Zhao2025RaSPD}. The median filter and ECA are purely physics-driven methods, whereas U-Net and Ra-SPD are deep learning methods. Specifically, U-Net is a purely data-driven baseline, while Ra-SPD integrates a spectral-spatial decomposition mechanism into its deep learning framework.}


\begin{figure*}[]
	\centerline{\includegraphics[width=0.95\linewidth]{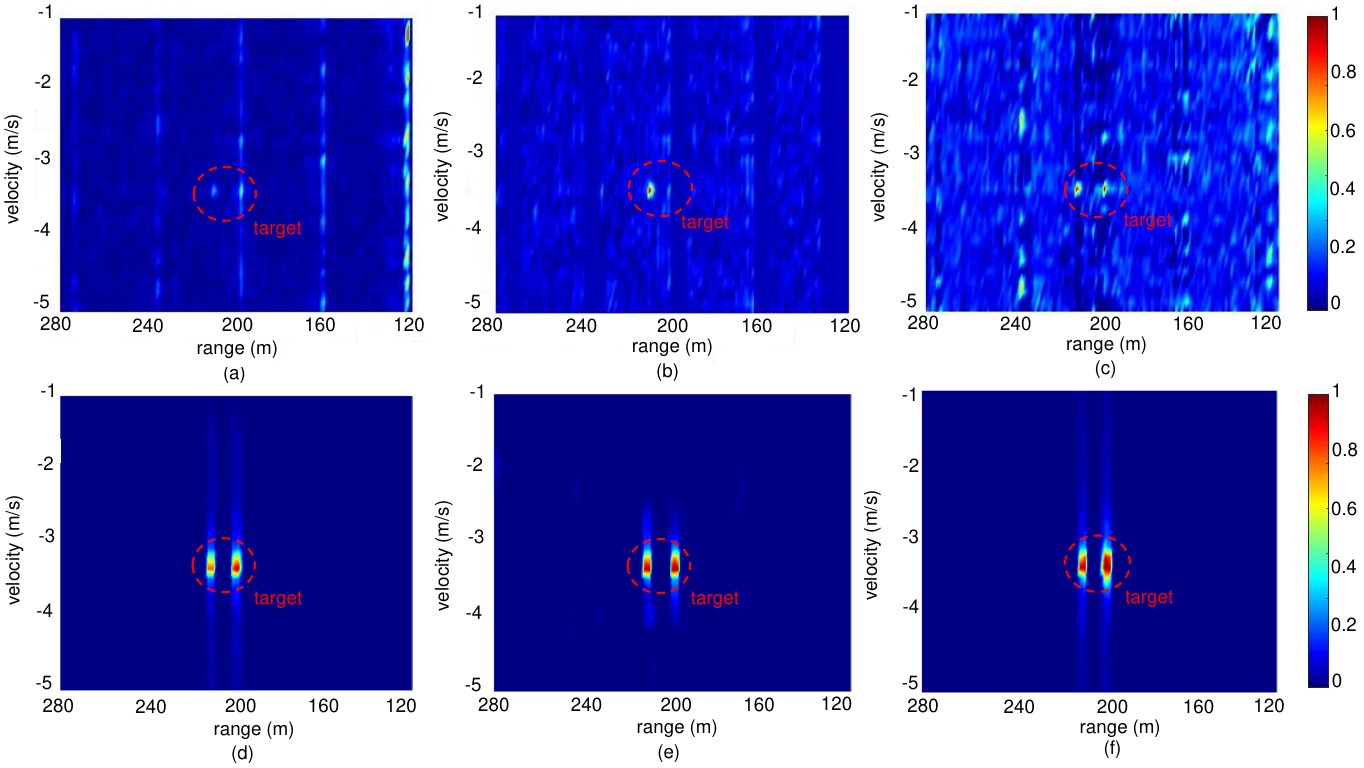}}
	\vspace{-0.3cm}
	\caption{\textmd{Normalized RV spectrum of the measured data before and after artifact elimination in {scenario B}. The target is obscured by the artifacts. (a) raw RV spectrum. (b) median filter. (c) ECA. (d) U-Net. (e) Ra-SPD. (f) the proposed PIAENet. The red ellipse indicates the target.}}\label{fig:simulation_result}
		\vspace{-0.5cm}
\end{figure*}

Fig. \ref{fig:simulation_result} presents the RV spectrum of the measured data, where the target is obscured by artifacts. 
According to the RV spectrum shown in Fig. \ref{fig:simulation_result}(b), when the target is obscured by artifacts, the median filter causes the target to be inadvertently removed as it eliminates those artifacts. {Fig. \ref{fig:simulation_result}(c) presents the results of ECA. As a typical physics-driven cancellation algorithm, ECA achieves moderate suppression of strong artifacts, but prominent residual striped artifacts still persist, leaving the target partially masked with limited enhancement in target detectability. The target peak remains sharp and distinct with a high signal-to-background ratio. Figs. \ref{fig:simulation_result}(d)-(f) show the output of U-Net, Ra-SPD, and the proposed PIAENet. In contrast to physics-driven methods, deep learning methods can preserve the target while simultaneously eliminating artifacts and noise.}

\begin{remark}
	It is worth noting that ISI was introduced due to the insufficient CP length. Therefore, although only a single UAV was used, two target points appeared in the RV spectrum. Since the objective of this paper is artifact mitigation, we do not address this phenomenon and treat the ground truth as corresponding to two targets.
\end{remark}

To further demonstrate the superiority of PIAENet, artifact elimination performance is evaluated by peak signal-to-noise ratio (PSNR), detection probability $P_{\mathrm{d}}$, and false target count with measured data from {scenario B}. The PSNR, which quantifies the similarity between the reconstructed RV spectrum $\hat{\boldsymbol{{\mathcal{X}}}}_i$ and the ground-truth spectrum ${\boldsymbol{{\mathcal{X}}}_i}^{(0)}$ of the $i$-th frame, is defined as
\begin{align}\label{eq:psnr}
	\text{PSNR}_i=10\lg\frac{\left(\max({\hat{\boldsymbol{\mathcal{X}}}_i})\right)^2}{{||\hat{\boldsymbol{{\mathcal{X}}}}_i-{\boldsymbol{{\mathcal{X}}}}_i^{(0)}||^2_2}/ K' L'}.
\end{align}
{In addition, detection probability $P_{\mathrm{d}}$ and false target count are evaluated by a cell-averaging constant false alarm rate (CA-CFAR) detector \cite{CFAR} with a fixed false alarm probability $P_{\mathrm{fa}}=10^{-3}$. The setting $P_{\mathrm{fa}}=10^{-3}$ is adopted as a representative CFAR operating point that balances detection probability and false target count. The guard window contains $3$ range cells and $3$ Doppler cells on each side of the cell under test, and the training window contains $9$ range cells and $18$ Doppler cells. A detected local peak is counted as a correct detection if it lies within a $\pm 2$-bin neighborhood of the RTK-derived target-related ground-truth peak.}

\begin{figure}[]
	\vspace{-0.2cm}
	\centerline{\includegraphics[width=1\linewidth]{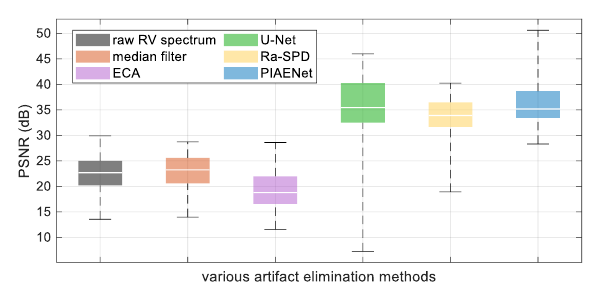}}
	\vspace{-0.5cm}
	\caption{{Boxplot of PSNR for RV spectrum subregions after processing by different artifact elimination methods.}}\label{fig:PSNR}
	\vspace{-0.3cm}
\end{figure}

\begin{figure}[]
	\centerline{\includegraphics[width=1\linewidth]{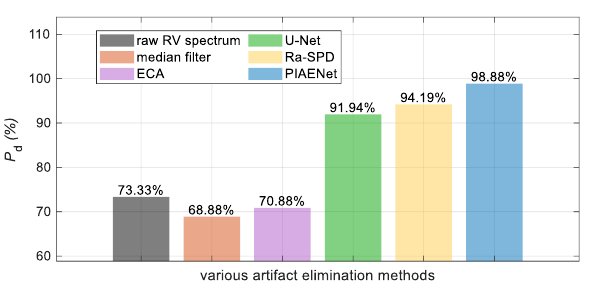}}
	\vspace{-0.3cm}
	\caption{{Detection probability $P_{\mathrm{d}}$ of CFAR ($P_{\mathrm{fa}}=10^{-3}$) for different artifact elimination methods.}}\label{fig:PD}
	\vspace{-0.6cm}
\end{figure}

\begin{figure}[]
	\vspace{-0.5em}
	\centerline{\includegraphics[width=1\linewidth]{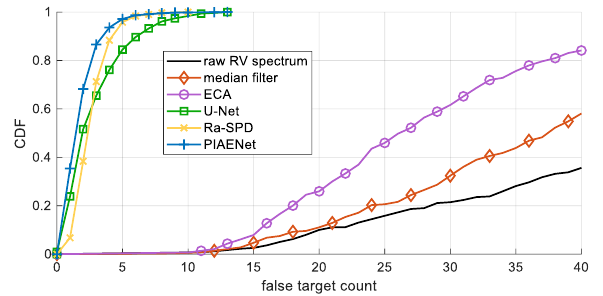}}
	\vspace{-0.3cm}
	\caption{{CDF of false target counts detected by CFAR ($P_{\mathrm{fa}}=10^{-3}$) for different artifact elimination methods.}}\label{fig:PfaN}
	\vspace{-0.5cm}
\end{figure}

Fig.~\ref{fig:PSNR} presents the boxplot of the PSNR for RV spectrum subregions after processing by different artifact elimination methods. {The raw RV spectrum without any artifact elimination processing achieves a mean PSNR of $22.53$~dB. The two traditional physics-driven methods, the median filter and ECA, bring only marginal improvement. The median filter has a median PSNR nearly identical to the raw spectrum, and ECA exhibits a slight decrease in the PSNR. In contrast, all three learning-based methods achieve a substantial leap in PSNR. The purely data-driven U-Net yields a mean PSNR of about $35.85$~dB, and Ra-SPD achieves a mean PSNR of around $34.08$~dB. The proposed PIAENet achieves the highest mean PSNR of $36.62$~dB, and its maximum PSNR exceeds $50.61$~dB.} This result validates that the proposed method can reconstruct the RV spectrum more accurately, outperforming both U-Net and Ra-SPD.

Fig. \ref{fig:PD} compares the target detection probability $P_{\mathrm{d}}$ of different methods. The raw RV spectrum achieves a baseline $P_{\mathrm{d}}$ of $73.33\%$, as the striped artifacts raise the noise floor and partially obscure the weak target, severely reducing target detectability. {Notably, the median filter and ECA degrade the detection performance compared with the raw spectrum.} All learning-based methods achieve a substantial improvement in $P_{\mathrm{d}}$. The proposed PIAENet achieves the highest $P_{\mathrm{d}}$ of $98.88\%$. This result fully demonstrates that the proposed PIAENet can effectively eliminate striped artifacts while preserving the complete target signal, significantly improving target detection performance even in scenarios where the target is heavily obscured by artifacts.

Fig. \ref{fig:PfaN} plots the cumulative distribution function (CDF) of false target counts. For CDF curves of false target counts, a curve shifted further toward the upper-left corner indicates superior false alarm suppression performance, as it means the vast majority of test samples have fewer false alarms.
The raw RV spectrum exhibits the worst false alarm performance, with the most right-shifted CDF curve. Its CDF value reaches merely $0.35$ when the false target count rises to $40$, meaning more than $65\%$ of the test samples have more than $40$ false targets. The median filter achieves only limited improvement over the raw spectrum, and $40\%$ of the test samples still have more than $40$ false targets. {ECA performs slightly better than the median filter, but it still leaves abundant residual artifacts that trigger false alarms.} In contrast, all learning-based methods achieve a dramatic leftward shift of the CDF curve, demonstrating significant superiority in false alarm suppression. Among all the compared methods, the proposed PIAENet achieves the best false alarm suppression performance. Specifically, the CDF of PIAENet reaches $0.95$ when the false target count is less than $5$, meaning $95\%$ of the test samples have false target counts controlled within $5$.

{

\begin{table}[]
	\vspace{-0em}
	\centering
	\caption{{Ablation study on the contribution of physical priors}}
	\label{table:ablation_study}
	\vspace{-0.5em}
	{
		\setlength{\tabcolsep}{5pt}
		\begin{tabular}{lcccc}
			\toprule
			& Config.~1 & Config.~2 & Config.~3 & Config.~4 \\
			\midrule
			selective mask & \ding{55} & \ding{51} & \ding{55} & \ding{51} \\& & \\[-5pt]
			multi-frame & \ding{55} & \ding{55} & \ding{51} & \ding{51} \\& & \\[-5pt]
			PSNR (dB) & 36.38 & 36.57 & 36.35 & 36.62 \\& & \\[-5pt]
			$P_{\mathrm{d}}$ ($\%$) & 86.97 & 90.81 & 97.58 & {98.88} \\& & \\[-5pt]
			$N_{\mathrm{fa}}$ & 7.14 & 3.48 & 3.29 & {2.45}\\
			\bottomrule
		\end{tabular}
	}
	\vspace{-1.5em}
\end{table}	

\begin{table*}[b]
	\centering
	\caption{Detection Performance Comparison of Different Methods Under Different ZC Root Indices}
	\label{tab:zc_root}
	\begin{tabular}{ccccccccccccc}
		\toprule
		\multirow{2}{*}{Methods} 
		& \multicolumn{2}{c}{$\mu=33$} 
		& \multicolumn{2}{c}{$\mu=51$} 
		& \multicolumn{2}{c}{$\mu=73$} 
		& \multicolumn{2}{c}{$\mu=97$} 
		& \multicolumn{2}{c}{$\mu=102$}
		& \multicolumn{2}{c} {overall}\\
		\cmidrule(lr){2-3} \cmidrule(lr){4-5} \cmidrule(lr){6-7} \cmidrule(lr){8-9} \cmidrule(lr){10-11} \cmidrule(lr){12-13}
		& $P_{\mathrm{d}}$ ($\%$) & $N_{\mathrm{fa}}$ 
		& $P_{\mathrm{d}}$ ($\%$) & $N_{\mathrm{fa}}$
		& $P_{\mathrm{d}}$ ($\%$) & $N_{\mathrm{fa}}$
		&$P_{\mathrm{d}}$ ($\%$) & $N_{\mathrm{fa}}$ 
		&$P_{\mathrm{d}}$ ($\%$) & $N_{\mathrm{fa}}$ 
		& $P_{\mathrm{d}}$ ($\%$) & $N_{\mathrm{fa}}$  \\
		\midrule
		raw RV spectrum & $52.85$& $66.32$ & $36.39$ & $63.98$ & $39.87 $ & $59.36$ & $20.57 $ & $71.16$ & $72.47 $ & $38.09$ & $44.43$ & $59.78$\\
		median filter   & $64.87$& $10.93$ & $65.19$& $10.46$ & $66.77$ & $15.78$ & $39.87 $ & $38.60$ & $72.47 $ & $14.34$ & $61.84$& $18.02$ \\
		ECA & $47.83$ & $20.09$ & $45.66$ & $17.20$ & $43.29$ & $12.92$ & $37.90$ & $16.02$ & $50.79$ & $10.91$ & $45.09$ & $15.43$\\
		U-Net & $95.25$ & $1.45$ & $95.25$ & $1.84$ & $93.99$ & $2.70$ & $92.41$ & $2.36$ & $98.10$ & $1.36$ & $95.00$ & $1.94$ \\
		Ra-SPD & $87.98$ & $3.77$ & $93.35$ & $0.89$ & $94.30$ & $0.92$ & $93.35$ & $0.78$ & $96.52$ & $0.67$ & $93.10$ & $1.41$\\
		PIAENet         & $99.21$ & $0.64$ & $98.95$ & $0.79$ & $99.54$ & $0.42$ & $99.54$ & $0.38$ & $99.47$ & $0.38$ &{$\mathbf{99.34}$} &\textbf{$\mathbf{0.52}$}\\
		\bottomrule
	\end{tabular}
\end{table*}

\vspace{-0.5em}
\subsection{Ablation Study}

We construct four configurations of the artifact eliminator by independently toggling the two proposed mechanisms, as shown in Table~\ref{table:ablation_study}. We compare the median PSNR, mean false target count $N_{\mathrm{fa}}$, and detection probability $P_{\mathrm{d}}$ of CFAR ($P_{\mathrm{fa}}=10^{-3}$). The selective mask alone increases the mean PSNR from $36.38$ to $36.57$~dB and improves $P_{\mathrm{d}}$ from $86.97\%$ to $90.81\%$, since it guides reconstruction toward artifact-prone bins. Multi-frame stacking alone also achieves a median PSNR of $36.35$~dB and raises $P_{\mathrm{d}}$ to $97.58\%$, confirming the value of temporal correlation. By combining both mechanisms, PIAENet achieves the highest $P_{\mathrm{d}}$ of $98.88\%$ and the lowest $N_{\mathrm{fa}}$ of $2.45$, while retaining a mean PSNR of $36.62$~dB. These results show that the selective mask identifies where reconstruction is needed, whereas the multi-frame mechanism supplies the temporal information required to suppress residual artifacts without removing weak targets.

\subsection{Robustness Validation via Simulation}

To supplement the real-world measurements, we perform comprehensive simulations to further evaluate the robustness of PIAENet under a wider range of operating conditions. In our simulation setup, the transmitting and receiving AAUs are positioned at $(0.0, 0.0, 0.0)$~m and $(7.2, 0.0, 0.0)$~m, respectively, while a static strong scatterer is placed at $(0.0, 20.0, 0.0)$~m. We introduce a CFO of $0.04$~ppm ($f_{\mathrm{o}}\approx 1$ kHz) to induce artifacts. As the simulated data are generated based on the theoretically derived model, they inherently preserve the same range extension and Doppler spreading characteristics as the measured artifacts, thereby ensuring consistency between the simulation and experimental results.

To evaluate the artifact elimination performance under varied conditions, we conduct simulations in which a target moves from 
$(0.0,560.0,0.0)$ m to $(0.0,610.0,0.0)$ m with a constant velocity of $(0.0,8.0,0.0)$ m/s. The performance of the competing methods is examined under five distinct ZC root indices, $\mu \in \{33, 51, 73, 97, 102\}$, each inducing a distinct artifact periodicity in the range domain according to \textit{Theorem}~\ref{th:period2}. The performance is assessed by detection probability $P_{\mathrm{d}}$ and mean false target count $N_{\mathrm{fa}}$ of CFAR ($P_{\mathrm{fa}}=10^{-3}$). As summarized in Table~\ref{tab:zc_root}, the results show that the proposed PIAENet maintains consistently high detection performance and low false alarm counts across all evaluated roots. This confirms that the proposed PIAENet generalizes well and remains effective across different ZC root indices.

We further evaluate PIAENet under five trajectory configurations to jointly probe its robustness to target range, velocity, and motion dynamics. In all cases, the ZC root index is fixed as $\mu=51$. Traj.~1 follows the original setting, where the target moves from the initial range to the terminal range with a fixed velocity of $8$~m/s. Traj.~2 reverses both the range evolution and the velocity direction, starting from the terminal range and moving back to the initial range with a fixed velocity of $-8$~m/s. Traj.~3 starts from the same initial range as Traj.~1, while its velocity continuously increases from $6.5$ to $9.5$~m/s, and the range evolution is calculated according to this time-varying velocity. Traj.~4 starts from the terminal range of the original trajectory and uses a negative velocity varying from $-6.5$ to $-9.5$~m/s, leading to a reversed range evolution. Traj.~5 starts from the same initial range as Traj.~1, with the velocity first increasing from $7$ to $9$~m/s and then decreasing from $9$ to $7$~m/s, representing non-uniform motion with a velocity reversal trend.
The trajectory-robustness table is prepared to report $P_{\mathrm{d}}$ and the mean false target count $N_{\mathrm{fa}}$ of CFAR ($P_{\mathrm{fa}}=10^{-3}$) under the five trajectory configurations. The numerical entries are temporarily left blank and will be filled after the corresponding experiments are finalized.

\begin{table*}[]
	\centering
	\caption{Detection Performance Comparison of Different Methods Under Different Target Trajectories}
	\begin{tabular}{ccccccccccccc}
		\toprule
		\multirow{2}{*}{Methods} 
		& \multicolumn{2}{c}{Traj.~1} 
		& \multicolumn{2}{c}{Traj.~2} 
		& \multicolumn{2}{c}{Traj.~3} 
		& \multicolumn{2}{c}{Traj.~4} 
		& \multicolumn{2}{c}{Traj.~5}
		& \multicolumn{2}{c} {overall}\\
		\cmidrule(lr){2-3} \cmidrule(lr){4-5} \cmidrule(lr){6-7} \cmidrule(lr){8-9} \cmidrule(lr){10-11} \cmidrule(lr){12-13}
		& $P_{\mathrm{d}}$ ($\%$) & $N_{\mathrm{fa}}$ 
		& $P_{\mathrm{d}}$ ($\%$) & $N_{\mathrm{fa}}$
		& $P_{\mathrm{d}}$ ($\%$) & $N_{\mathrm{fa}}$
		& $P_{\mathrm{d}}$ ($\%$) & $N_{\mathrm{fa}}$ 
		& $P_{\mathrm{d}}$ ($\%$) & $N_{\mathrm{fa}}$ 
		& $P_{\mathrm{d}}$ ($\%$) & $N_{\mathrm{fa}}$  \\
		\midrule
		raw RV spectrum & $36.08$ & $63.90$ & $36.39$ & $64.73$ & $59.49$ & $63.52$ & $62.34$ & $64.19$ & $62.34$ & $64.18$ & $51.33$ & $64.11$\\
		median filter & $64.24$ & $10.67$ & $60.76$ & $10.98$ & $73.73$ & $10.63$ & $74.68$ & $10.50$ & $77.85$ & $10.83$ & $70.25$ & $10.72$\\
		ECA & $47.17$ & $17.21$ & $42.96$ & $16.82$ & $61.78$ & $17.63$ & $58.62$ & $17.64$ & $59.47$ & $17.55$ & $54.00$ & $17.37$\\
		U-Net & $85.44$ & $0.21$ & $84.18$ & $0.16$ & $90.51$ & $0.21$ & $91.14$ & $0.18$ & $86.39$ & $0.22$ & $87.53$ & $0.19$\\
		Ra-SPD & $94.94$ & $1.84$ & $91.46$ & $2.46$ & $93.67$ & $2.23$ & $95.89$ & $2.10$ & $93.67$ & $1.96$ & $93.92$ & $2.12$\\
		PIAENet & $99.21$ & $1.51$ & $99.54$ & $1.44$ & $95.13$ & $1.84$ & $97.83$ & $1.85$ & $93.36$ & $1.86$ & {$\mathbf{97.01}$} & {$\mathbf{1.70}$}\\
		\bottomrule
	\end{tabular}
\end{table*}



Taken together, these simulation results, grounded in the theoretically validated artifact model, provide broader evidence that the robustness of PIAENet extends beyond the specific ZC root and target trajectory conditions.

}

\section{Conclusion}\label{sec:conclusion}
This paper focuses on the observation, modeling, and elimination of polyphase-code artifacts in 5G NR-based sensing. 
Initially, we identify the presence of striped artifacts in the RV spectrum when polyphase-code sequences are used as reference signals. These artifacts severely impair weak target detection and increase false alarms. 
On this basis, we establish a theoretical model that explains the physical origin of the artifacts. The analysis reveals that these artifacts exhibit periodic extensions in the range domain and spectral spreading in the velocity domain. 
Subsequently, based on the characteristics of the artifacts, we propose a data-physics-driven framework, PIAENet, centered on a multi-frame selective-masked encoder-decoder. This framework incorporates temporal continuity and spatial consistency priors into network training, enabling effective artifact elimination while maintaining high sensitivity to weak targets. 
Ultimately, we comprehensively validate the proposed method through real-world measurements. Experimental results show that, compared with traditional physics-driven methods and purely data-driven approaches, PIAENet significantly improves target detection performance and reduces false alarms, even in harsh scenarios where weak targets are obscured by artifacts.

Future research will be extended in three promising directions. First, we will extend the application of PIAENet to more diverse scenarios and evaluate its potential for practical deployment. Second, we will optimize polyphase-code sequences and design alternative orthogonal sequences with favorable correlation properties (e.g., pseudorandom codes) to replace ZC sequences. Finally, we will optimize the transmission waveform by developing advanced ISAC waveforms with robust anti-interference performance to replace OFDM.

\vspace{-0.5em}
\appendices
{\section{Velocity/Doppler Domain Spreading Derivation}\label{AppendixA}
Assume that for $n=0,...,N-1$, the equality $g_l[n]=g^{}_l$ holds. For all $n$, the time-domain signal of the $l$-th symbol with multiplicative error is 
\begin{align}\label{appex:appexlA0k}
	\tilde{y}_l={g}^{}_l{y}_l,\quad l=0,...,L-1,
\end{align}
where $\mathbf{g}^{}=[{g}^{}_0,...,{g}^{}_{L-1}]^{\mathsf{T}}$ and $\mathbf{y}=[y_0,...,y_{L-1}]^{\mathsf{T}}$ are the multiplicative noise and the ideal pilot signal, respectively. Here, we proceed by ignoring the scattering coefficient ($\beta = 1$) and delay ($\tau=0$) for simplicity. After the time-frequency transformation and pilot unloading as (\ref{eq:xlq1})-(\ref{eq:sensingsignalH1}), the CSI obtained from the pilot signal of all $k$ for a single path is derived as
\begin{align}\label{appex:appexlA1k}
	{\tilde{h}}_l={g}^{}_l{{h}}_l,
\end{align}
where ${{h}}_l=\mathrm{e}^{\jmath 2\mathrm{\uppi}f_{\mathrm{d},0}lT}$, $f_{\mathrm{d},0}$, and $T$ are the Doppler and total symbol duration. We further focus on the non-ideal impact on Doppler estimation. By Fourier transform, we have the ideal and non-ideal Doppler spectra as
\begin{subequations}\label{appex:appexlA2k}
\begin{align}
	{H}[\upsilon]=\mathcal{F}\{{h}_l\}=\sum_{l=0}^{L-1} \mathrm{e}^{\mathrm{j} 2\mathrm{\uppi}(f_{\mathrm{d},0}-f_{\mathrm{d}}[\upsilon])lT},
\end{align}
\begin{align}
	&\tilde{{H}}[\upsilon]=
	\mathcal{F}\{{\tilde{h}}_l\}=\frac{1}{L}{G}^{}[\upsilon]\ast{H}[\upsilon]
	\\&=\frac{1}{L}\sum_{\iota=-L+1}^{L-1}{G}^{}[\iota]{H}[\upsilon-\iota],\nonumber
\end{align}
\end{subequations}
where $f_{\mathrm{d}}[\upsilon]$ is the Doppler and $\upsilon$ is the corresponding index in the Doppler domain. ${G}^{}[\upsilon]$ is the DFT of ${g}^{}_l$. This implies that the ideal spectrum concentrated at $f_{\mathrm{d},0}$ undergoes convolution with $G^{}[\upsilon]$, resulting in spreading across other Doppler bins.

\section{Range/Delay domain Extension Derivation}\label{AppendixB}
For all $l$, the time-domain signal of the $n$-th sampling point with multiplicative error is 
\begin{align}\label{appex:appexlB0k}
	\tilde{y}[n]={g}^{}[n]{y}[n],\quad n=0,...,N-1,
\end{align}
where $\mathbf{g}^{}=[{g}^{}[0],...,{g}^{}{[N-1]}]^{\mathsf{T}}$ and $\mathbf{y}=[{y}^{}[0],...,{y}^{}{[N-1]}]^{\mathsf{T}}$ are the multiplicative noise and the ideal pilot signal, respectively. Here, we proceed by ignoring the scattering coefficient ($\beta = 1$) and Doppler ($f_{\mathrm{d}}=0$) for simplicity.
Thus, the ideal and non-ideal pilot signals for a single path in the frequency domain are 
\begin{subequations}\label{appexeq:xl1k}
	\begin{align}
		{Y}_k=\mathcal{F}\{{y}[n]\}=S_k\mathrm{e}^{-\mathrm{j} 2\mathrm{\uppi}\tau_{0}k\Delta_{\mathrm{f}}},\quad k=0,..,K-1,
	\end{align}
	\begin{align}
		\tilde{Y}_k=\mathcal{F}\{{g}^{}[n]{y}[n]\}=\frac{1}{N}{G}^{}_k\ast{Y}_k
		=\frac{1}{N}\sum_{m=-K+1}^{K-1}{G}^{}_m{Y}_{k-m},
	\end{align}
\end{subequations}
where $\tau_0$ and $\Delta_\mathrm{f}$ are the delay and subcarrier spacing, $S_k=a_k\mathrm{e}^{\mathrm{j} 2\mathrm{\uppi}\phi_k}$ is the pilot sequence, ${G}^{}_k$ is the DFT of ${g}^{}[n]$. Here $N=K$. In (\ref{appexeq:xl1k}b), all terms with $m\neq0$ are treated as ICI. After pilot unloading as (\ref{eq:sensingsignalH1}), the CSI obtained from non-ideal pilot signal in (\ref{appexeq:xl1k}b) is derived as
\begin{align}\label{appeq:hq0}
	&\tilde{H}_k=\tilde{Y}_k/S_k
	=\frac{1}{N}\left(\sum_{m=-K+1}^{K-1}{G}^{}_k{Y}_{k-m}/S_k\right).
\end{align}
We further focus on the non-ideal impact on delay estimation. By inverse Fourier transform for $\tilde{H}_k$ in (\ref{appeq:hq0}), we have the non-ideal delay spectrum as
\begin{align}\label{appeq:hq0d}
	\tilde{R}[d]=\frac{1}{N}\left(\sum_{m=0}^{K-1}\left({G}^{}_m\sum_{k=m}^{K-1}\frac{S_{k-m}}{S_{k}}\mathrm{e}^{\mathrm{j} 2\mathrm{\uppi}(\tau[d]-\tau_{0})(k-m)\Delta_{\mathrm{f}}}\right)\right.\\\left.+\sum_{m=-K+1}^{-1}\left({G}^{}_m\sum_{k=0}^{K+m-1}\frac{S_{k-m}}{S_{k}}\mathrm{e}^{\mathrm{j} 2\mathrm{\uppi}(\tau[d]-\tau_{0})(k-m)\Delta_{\mathrm{f}}}\right)\right).\nonumber
\end{align}
where $\tau[d]$ is the delay and $d$ is the corresponding index of delay domain.
In (\ref{appeq:hq0d}), the sum of $m$-th and $(-m)$-th term ($m\neq0$) is induced by the non-ideal impairments, which is written as
\begin{align}\label{appeq:Xi}
	&\xi_{m}[d]=\frac{1}{N}\left({G}^{}_m\sum_{k=m}^{K-1}\frac{S_{k-m}}{S_{k}}\mathrm{e}^{\mathrm{j} 2\mathrm{\uppi}(\tau[d]-\tau_{0})(k-m)\Delta_{\mathrm{f}}}+\right.\\&\left.{G}^{}_{-m}\sum_{k=0}^{K-m-1}\frac{S_{k+m}}{S_{k}}\mathrm{e}^{\mathrm{j} 2\mathrm{\uppi}(\tau[d]-\tau_{0})(k+m)\Delta_{\mathrm{f}}}\right), m=1,...,K-1.\nonumber
\end{align}
{
Consider a single-index quadratic-phase polyphase pilot
\begin{equation}
S_k=\mathrm{e}^{\mathrm{j}\left(ak^2+bk+c\right)},
\, k=0,\ldots,K-1,
\label{appeq:quadratic_general}
\end{equation}
where $a$, $b$, and $c$ are real constants. This class includes ZC, P3, and P4 codes. For an ICI order $m>0$, its
two pilot ratios are
\begin{align}
\frac{S_{k-m}}{S_k}
=\mathrm{e}^{\mathrm{j}\left(am^2-bm\right)}
\mathrm{e}^{-\mathrm{j}2amk},
\frac{S_{k+m}}{S_k}
=\mathrm{e}^{\mathrm{j}\left(am^2+bm\right)}
\mathrm{e}^{\mathrm{j}2amk}.
\label{appeq:quadratic_ratio}
\end{align}
Thus, both ratios are single linear-phase terms in $k$. Define the signed
delay shift
\begin{equation}
\delta_{\tau}^{\mathrm Q}[m]=\frac{am}{\pi\Delta_{\mathrm f}}.
\label{appeq:quadratic_delay_shift}
\end{equation}
Let $
C_m^{\mathrm Q,\mp}[d]=
\mathrm{e}^{\mathrm{j}(am^2\mp bm)\mp
\mathrm{j}2\pi(\tau[d]-\tau_0)m\Delta_{\mathrm f}}$,
which is independent of $k$. Substituting \eqref{appeq:quadratic_ratio} into
\eqref{appeq:Xi} gives
\begin{align}
\xi_m^{\mathrm Q}[d]
={}&\frac{1}{N}\Bigg(
G_m C_m^{\mathrm Q,-}[d]
\sum_{k=m}^{K-1}
\mathrm{e}^{\mathrm{j}2\pi\left(\tau[d]-\tau_0-\delta_{\tau}^{\mathrm Q}[m]\right)k\Delta_{\mathrm f}}
\nonumber\\
&+G_{-m}C_m^{\mathrm Q,+}[d]
\sum_{k=0}^{K-m-1}
\mathrm{e}^{\mathrm{j}2\pi\left(\tau[d]-\tau_0+\delta_{\tau}^{\mathrm Q}[m]\right)k\Delta_{\mathrm f}}
\Bigg).
\label{appeq:quadratic_xi}
\end{align}
The factors $C_m^{\mathrm Q,\mp}[d]$ affect only the complex
amplitudes. Hence, the $m$-th ICI contribution produces a pair of
sharp delay-domain kernels centered at $\tau_0\pm\delta_{\tau}^{\mathrm Q}[m]$.
The corresponding magnitude of the range offset is
\begin{equation}
\left|\Delta d_m^{\mathrm Q}\right|
=\frac{c_0|a|m}{2\pi\Delta_{\mathrm f}}.
\label{appeq:quadratic_range_shift}
\end{equation}
}
As a ZC example, substituting $S_k^{\mathrm{ZC}}=\mathrm{e}^{-\mathrm{j}\pi\mu k(k+1)/K}$
into \eqref{appeq:quadratic_general} gives
\begin{equation}
a=-\frac{\pi\mu}{K},\qquad b=-\frac{\pi\mu}{K},\qquad c=0.
\end{equation}
Hence, the two-sided ZC artifact locations are
\begin{equation}
\tau[d]=\tau_0\pm\delta_{\tau}^{\mathrm{ZC}}[m],
\qquad
\delta_{\tau}^{\mathrm{ZC}}[m]=\frac{\mu m}{K\Delta_{\mathrm f}},
\label{appeq:deltad}
\end{equation}
where the sign convention only exchanges the two ICI contributions. The
corresponding range offsets and the spacing between adjacent ICI orders are
\begin{equation}
\left|\Delta d_m^{\mathrm{ZC}}\right|
=\frac{c_0\mu m}{2K\Delta_{\mathrm f}},
\qquad
\Delta R_{\mathrm{ZC}}=\frac{c_0\mu}{2K\Delta_{\mathrm f}},
\label{appeq:zc_period}
\end{equation}
where $c_0$ is the speed of light.

{
\section{Cyclic Matrix-Polyphase-Code Range Extension}\label{AppendixC}

Frank, P1, and P2 are matrix-polyphase codes with an $M\times M$ phase
table and code length $N_{\mathrm C}=M^2$. Let $\kappa_k\equiv k\pmod
{N_{\mathrm C}}$ and define
\begin{equation}\label{appeq:matrix_general}
I_k=\left\lfloor\frac{\kappa_k}{M}\right\rfloor,\,
J_k=\kappa_k-MI_k,\,
S_k^{\mathrm C}=\mathrm{e}^{\mathrm{j}\Phi(I_k,J_k)}.
\end{equation}
Here, $\Phi(i,j)$ is the phase table of the selected code. In particular,
\begin{equation}\label{appeq:matrix_examples}
\Phi_{\mathrm F}(i,j)=\frac{2\pi ij}{M},\,
\Phi_{\mathrm{P1}}(i,j)=-\frac{\pi}{M}(M-2i-1)(iM+j),
\end{equation}
where the latter is the standard P1 definition for even $M$. A cyclic code-to-subcarrier mapping is defined 
for an arbitrary number $K$ of occupied subcarriers.
For a fixed ICI order $m$, the cyclic mapping in
\eqref{appeq:matrix_general} makes the pilot-ratio sequences
$N_{\mathrm C}$-periodic. Re-index the negative-ICI and positive-ICI sums in
\eqref{appeq:Xi} by $u=k-m$ and $v=k+m$, respectively, and define the
$N_{\mathrm C}$-point Fourier-series coefficients
\begin{align}
A_{\ell,m}^{\mathrm C,-}
&=\frac{1}{N_{\mathrm C}}\sum_{u=0}^{N_{\mathrm C}-1}
\frac{S_u^{\mathrm C}}{S_{u+m}^{\mathrm C}}
\mathrm{e}^{\mathrm{j}2\pi\ell u/N_{\mathrm C}},\\
A_{\ell,m}^{\mathrm C,+}
&=\frac{1}{N_{\mathrm C}}\sum_{v=0}^{N_{\mathrm C}-1}
\frac{S_v^{\mathrm C}}{S_{v-m}^{\mathrm C}}
\mathrm{e}^{-\mathrm{j}2\pi\ell v/N_{\mathrm C}}.
\end{align}
Here, $\ell=0,\ldots,N_{\mathrm C}-1$, and every pilot index is evaluated
using the cyclic mapping in \eqref{appeq:matrix_general}. Hence,
\begin{align}
\frac{S_u^{\mathrm C}}{S_{u+m}^{\mathrm C}}
&=\sum_{\ell=0}^{N_{\mathrm C}-1}
A_{\ell,m}^{\mathrm C,-}\mathrm{e}^{-\mathrm{j}2\pi\ell u/N_{\mathrm C}},\\
\frac{S_v^{\mathrm C}}{S_{v-m}^{\mathrm C}}
&=\sum_{\ell=0}^{N_{\mathrm C}-1}
A_{\ell,m}^{\mathrm C,+}\mathrm{e}^{\mathrm{j}2\pi\ell v/N_{\mathrm C}}.
\end{align}
Substituting these expansions into \eqref{appeq:Xi} gives
\begin{align}
\xi_m^{\mathrm C}[d]
={}&\frac{1}{N}\Bigg(
G_m\sum_{\ell=0}^{N_{\mathrm C}-1}A_{\ell,m}^{\mathrm C,-}
\sum_{u=0}^{K-m-1}
\mathrm{e}^{\mathrm{j}2\pi\left(\tau[d]-\tau_0-
\frac{\ell}{N_{\mathrm C}\Delta_{\mathrm f}}\right)u\Delta_{\mathrm f}}
\nonumber\\
&+G_{-m}\sum_{\ell=0}^{N_{\mathrm C}-1}A_{\ell,m}^{\mathrm C,+}
\sum_{v=m}^{K-1}
\mathrm{e}^{\mathrm{j}2\pi\left(\tau[d]-\tau_0+
\frac{\ell}{N_{\mathrm C}\Delta_{\mathrm f}}\right)v\Delta_{\mathrm f}}
\Bigg),
\label{appeq:matrix_local_xi}
\end{align}
which has the same shifted-kernel form as
\eqref{appeq:quadratic_xi}. The candidate delay shifts are
\begin{equation}
\delta_{\tau}^{\mathrm C}[\ell]
=\frac{\ell}{N_{\mathrm C}\Delta_{\mathrm f}},
\qquad \ell=0,\ldots,N_{\mathrm C}-1.
\label{appeq:matrix_delay_grid}
\end{equation}
The coefficients $A_{\ell,m}^{\mathrm C,-}$ and
$A_{\ell,m}^{\mathrm C,+}$ determine which candidate shifts are visible and
their amplitudes. Thus, Frank, P1, and P2 codes share the same candidate
delay grid under a common cyclic $M\times M$ mapping, but have different
artifact amplitudes because their phase tables are different. When
$K<N_{\mathrm C}$, the finite observation window has a delay resolution of
approximately $1/(K\Delta_{\mathrm f})$, which is coarser than the candidate
grid spacing $1/(N_{\mathrm C}\Delta_{\mathrm f})$.
}

}


 

\begin{thebibliography}{1}
\bibliographystyle{IEEEtran}

	\bibitem{You2023Toward}
	X. You, Y. Huang, S. Liu, D. Wang, J. Ma, C. Zhang, H. Zhan, C. Zhang, J. Zhang, Z. Liu \textit{et al.}, ``Toward 6G TK${\upmu}$ extreme connectivity: Architecture, key technologies and experiments,'' \textit{IEEE Wireless Commun.}, vol. 30, no. 3, pp. 86–95, Jun. 2023.
	
	\bibitem{ITU2022Future}
	ITU-R, ``Future technology trends of terrestrial International Mobile Telecommunications systems towards 2030 and beyond,'' ITU-R Report M.2516-0, 2022.
	
	\bibitem{Liu2022Integrated}
	F. Liu, Y. Cui, C. Masouros, J. Xu, T. X. Han, Y. C. Eldar, and S. Buzzi, ``Integrated sensing and communications: Toward dual-functional wireless networks for 6G and beyond,'' \textit{IEEE J. Sel. Areas Commun.}, vol. 40, no. 6, pp. 1728–1767, Jun. 2022.
	
	\bibitem{Nuria2024integrated}
	N. González-Prelcic, M. F. Keskin, O. Kaltiokallio, M. Valkama, D. Dardari, X. Shen, Y. Shen, M. Bayraktar, and H. Wymeersch, ``The integrated sensing and communication revolution for 6G: Vision, techniques, and applications,'' \textit{Proc. IEEE}, vol. 112, no. 7, pp. 676–723, Jul. 2024.
	
	\bibitem{3GPPISAC1}
	3GPP, ``Feasibility study on integrated sensing and communication,'' 3GPP TR 22.837, Release 19, 2024.
	
	\bibitem{3GPPISAC2}
	3GPP, ``Study on Integrated Sensing And Communication (ISAC) for NR,'' 3GPP TR 38.765, Release 20, 2026.
	
	
	
	\bibitem{Shi2022Device}
	Q. Shi, L. Liu, S. Zhang, and S. Cui, ``Device-free sensing in OFDM cellular network,'' \textit{IEEE J. Sel. Areas Commun.}, vol. 40, no. 6, pp. 1838–1853, Jun. 2022.
	\bibitem{Yang2026Hierarchical}
	R. Yang, X. Li, Y. Huang, L. Yang and W. Zhang, ``Hierarchical reinforcement learning-based beam selection for integrated sensing and communication systems,'' \textit{IEEE Trans. Wireless Commun.}, vol. 25, pp. 1767-1780, 2026.
	
	\bibitem{Wypich2025passive}
	M. Wypich, R. Maksymiuk and T. P. Zielinski, "5G-based passive radar utilizing channel response estimated via reference signals," \textit{IEEE Trans. on Radar Syst.}, vol. 3, pp. 511-519, 2025.
	
	\bibitem{Zhang2025Target}
	Z. Zhang, H. Ren, C. Pan, S. Hong, D. Wang, J. Wang and X. You, ``Target localization in cooperative ISAC systems: A Scheme based on 5G NR OFDM signals,'' \textit{IEEE Trans. Commun.}, vol. 73, no. 5, pp. 3562-3578, May 2025.
	
	
	\bibitem{Chu1972Polyphase}
	D. Chu, ``Polyphase codes with good periodic correlation properties (Corresp.),'' \textit{IEEE Trans. Inf. Theory}, vol. 18, no. 4, pp. 531–532, Jul. 1972.
	\bibitem{Liu2026Cooperative}
	H. Liu, Z. Wei, L. Sun, R. Xu, Y. Zhang and Z. Feng, ``Cooperative sensing in cell-free massive MIMO ISAC systems: Performance optimization and signal processing,'' \textit{IEEE Trans. Wireless Commun.}, vol. 25, pp. 12531-12547, 2026.
	
	\bibitem{Na2025Integrated}
	N. K. Nataraja, S. Sharma, K. Ali, F. Bai, R. Wang and A. F. Molisch, ``Integrated sensing and communication (ISAC) for vehicles: Bistatic radar with 5G-NR signals,'' in \textit{IEEE Trans. Veh. Technol.}, vol. 74, no. 4, pp. 6121-6137, Apr. 2025.
	
	\bibitem{Ding2025Bi}
	S. Ding, J. Li, B. Chen, D. Jiang, J. Yao, F. Qin, J. Tan, Y. Yuan, D. Zhang, and C.-L. I, ``Bi-static ISAC with asynchronous transceivers: Mechanism, solution, and field test,'' \textit{IEEE Internet Things J.}, vol. 12, no. 17, pp. 35923–35940, Sep. 2025.
	
	\bibitem{Mao2025Model}
	S. Liu, Z. Mao, X. Li, M. Pan, P. Liu, Y. Huang and X. You, ``Model-driven deep neural network for enhancing direction finding with commodity 5G gNodeB,'' \textit{ACM Trans. Sens. Netw.}, vol. 21, no. 2, pp. 1-25, Mar. 2025.
	\bibitem{Zhang2024Dual}
	H. Zhang, S. Wei, X. Cai, L. Nie, M. Wang, J. Shi, and G. Cui, ``Dual-domain feature-oriented interference suppression for FMCW automotive radar,'' \textit{IEEE Sens. J.}, vol. 24, no. 5, pp. 6405–6417, Mar. 2024.
	
	\bibitem{Lee2021Mutual}
	S. Lee, J.-Y. Lee, and S.-C. Kim, ``Mutual interference suppression using wavelet denoising in automotive FMCW radar systems,'' \textit{IEEE Trans. Intell. Transp. Syst.}, vol. 22, no. 2, pp. 887–897, Feb. 2021.
	
	\bibitem{Xu2021Interference}
	Z. Xu and M. Yuan, ``An interference mitigation technique for automotive millimeter wave radars in the tunable Q-Factor wavelet transform domain,'' \textit{IEEE Trans. Microw. Theory Techn.}, vol. 69, no. 12, pp. 5270–5283, Dec. 2021.
	
	%
	%
	
	\bibitem{Alland2019Interference}
	S. Alland, W. Stark, M. Ali, and M. Hegde, ``Interference in automotive radar systems: Characteristics, mitigation techniques, and current and future research,'' \textit{IEEE Signal Process. Mag.}, vol. 36, no. 5, pp. 45–59, Sep. 2019.
	
	\bibitem{Jin2019Automotive}
	F. Jin and S. Cao, ``Automotive radar interference mitigation using adaptive noise canceller,'' \textit{IEEE Trans. Veh. Technol.}, vol. 68, no. 4, pp. 3747–3754, Apr. 2019.
	
	
	\bibitem{Baral2023Automotive}
	A. B. Baral, B. R. Upadhyay, and M. Torlak, ``Automotive radar interference mitigation using two-stage signal decomposition approach,'' in \textit{Proc. IEEE Radar Conf. (RadarConf23)}, San Antonio, TX, USA, 2023, pp. 1–6.
	
	\bibitem{Siam2025Artificial}
	S. I. Siam, H. Ahn, L. Liu, S. Alam, H. Shen, Z. Cao, N. Shroff, B. Krishnamachari, M. Srivastava, and M. Zhang, ``Artificial intelligence of things: A survey,'' \textit{ACM Trans. Sens. Netw.}, vol. 21, no. 1, pp. 1–75, Jan. 2025.
	
	\bibitem{Nirmal2021Deep}
	I. Nirmal, A. Khamis, M. Hassan, W. Hu, and X. Zhu, ``Deep learning for radio-based human sensing: Recent advances and future directions,'' \textit{IEEE Commun. Surveys Tuts.}, vol. 23, no. 2, pp. 995–1019, Jun. 2021.
	
	\bibitem{Ristea2020Fully}
	N.-C. Ristea, A. Anghel, and R. T. Ionescu, ``Fully convolutional neural networks for automotive radar interference mitigation,'' in \textit{Proc. IEEE 92nd Veh. Technol. Conf. (VTC2020-Fall)}, Victoria, BC, Canada, 2020, pp. 1–5.
	
	\bibitem{Fuchs2020Automotive}
	J. Fuchs, A. Dubey, M. Lübke, R. Weigel, and F. Lurz, ``Automotive radar interference mitigation using a convolutional autoencoder,'' in \textit{Proc. IEEE Int. Radar Conf. (RADAR)}, Washington, DC, USA, 2020, pp. 1–6.
	
	\bibitem{Zhu2021Low}
	L. Zhu, S. Zhang, K. Chen, S. Chen, X. Wang, D. Wei, and H. Zhao, ``Low-SNR recognition of UAV-to-ground targets based on micro-Doppler signatures using deep convolutional denoising encoders and deep residual learning,'' \textit{IEEE Trans. Geosci. Remote Sens.}, vol. 60, pp. 1–13, 2021, Art. no. 5103913.
	
	
	\bibitem{Wang2024Interference}
	J. Wang, R. Li, X. Zhang, and Y. He, ``Interference mitigation for automotive FMCW radar based on contrastive learning with dilated convolution,'' \textit{IEEE Trans. Intell. Transp. Syst.}, vol. 25, no. 1, pp. 545–558, Jan. 2024.
	
	%
	
	\bibitem{Wang2025Physics}
	R. Wang and R. Yu, ``Physics-guided deep learning for dynamical systems: A survey,'' \textit{ACM Comput. Surv.}, vol. 58, no. 5, pp. 1–31, Jun. 2025.
	
	\bibitem{Wang2022Prior}
	J. Wang, R. Li, Y. He, and Y. Yang, ``Prior-guided deep interference mitigation for FMCW radars,'' \textit{IEEE Trans. Geosci. Remote Sens.}, vol. 60, pp. 1–16, 2022, Art. no. 5111216.
	
	\bibitem{Park2024Interference}
	D.-H. Park, G.-H. Park, J.-H. Park, J.-H. Bang, D. Kim, and H.-N. Kim, ``Interference suppression for an FM-radio-based passive radar via deep convolutional autoencoder,'' \textit{IEEE Trans. Aerosp. Electron. Syst.}, vol. 60, no. 1, pp. 106–118, Feb. 2024.
	
	\bibitem{Pang2025MFS}
	X. Pang, Y. Peng, P. Wang, W. Wang, and W. Xiang, ``MFS: A motion feature separation model for UAV detection under passive radar,'' \textit{IEEE Trans. Aerosp. Electron. Syst.}, vol. 61, no. 5, pp. 11450–11468, Oct. 2025.
	
	\bibitem{Zhang2024FUAS}
	H. Zhang, S. Wei, M. Wang, Y. Hu, J. Shi, and G. Cui, ``FUAS-Net: Feature-oriented unsupervised network for FMCW radar interference suppression,'' \textit{IEEE Trans. Microw. Theory Techn.}, vol. 72, no. 4, pp. 2602–2619, Apr. 2024.
	
	\bibitem{Zhao2025RaSPD}
	D. Zhao, Y. Yang, J. Yang, B. Li, Y. He, and Y. Lang, ``Ra-SPD: Radar signal interference mitigation using spectral–spatial decomposition,'' \textit{IEEE Internet of Things J.}, vol. 12, no. 22, pp. 47348–47365, Nov. 2025.
	
	\bibitem{Khalid2019Convolutional}
	H. Khalid, S. Pollin, M. Rykunov, A. Bourdouxc, and H. Sahli, ``Convolutional long short-term memory networks for doppler-radar based target classification,'' in \textit{Proc. IEEE Radar Conf. (RadarConf)}, Boston, MA, USA, 2019, pp. 1–6.
	
	\bibitem{3GPPNR}
	3GPP, ``NR; Physical channels and modulation,'' 3GPP TS 38.211, Release 19, 2025.
	\bibitem{CNNclass}
	Y. LeCun, L. Bottou, Y. Bengio, and P. Haffner, ``Gradient-based learning applied to document recognition,'' \textit{Proc. IEEE}, vol. 86, no. 11, pp. 2278–2324, Nov. 1998.
	\bibitem{UNet}
	O. Ronneberger, P. Fischer, and T. Brox, ``U-Net: Convolutional networks for biomedical image segmentation,'' in \textit{Proc. Int. Conf. Med. Image Comput. Comput.-Assist. Intervent. (MICCAI)}, Munich, Germany, 2015, pp. 234–241.
	\bibitem{CFAR}H. Rohling, ``Radar CFAR thresholding in clutter and multiple target situations,'' \textit{IEEE Trans. Aerosp. Electron. Syst.}, vol. 19, no. 4, pp. 608–621, Jul. 1983.

\end{thebibliography}
%
\makeatletter
\renewcommand{\@biblabel}[1]{\makebox[1.2em][l]{[#1]}} 
\def\IEEEbiblabelsep{0.5em} 
\makeatother

%
 




\vfill

\end{document}